%% file: distancelowrank_arxiv.tex
\documentclass[11pt]{article}

\usepackage{amssymb,amsmath,amsthm,sectsty,url}
\usepackage[letterpaper,hmargin=1.0in,vmargin=1.0in]{geometry}
\usepackage[pdftex,colorlinks,linkcolor=blue,citecolor=blue,filecolor=blue,urlcolor=blue]{hyperref}
\usepackage{cleveref}
\usepackage{color}
\usepackage[boxed]{algorithm}
\usepackage{tikz}
\usepackage{graphicx}
\usepackage{nicefrac}
\usepackage{aliascnt}
\usepackage{algorithmic}

\newtheorem{theorem}{Theorem}[section]
\crefname{theorem}{Theorem}{Theorems}

\newaliascnt{lemma}{theorem}
\newtheorem{lemma}[lemma]{Lemma}
\aliascntresetthe{lemma}
\crefname{lemma}{Lemma}{Lemmas}

\newaliascnt{proposition}{theorem}

\aliascntresetthe{proposition}
\crefname{proposition}{Proposition}{Propositions}

\newaliascnt{corollary}{theorem}
\newtheorem{corollary}[corollary]{Corollary}
\aliascntresetthe{corollary}
\crefname{corollary}{Corollary}{Corollaries}

\newaliascnt{fact}{theorem}

\aliascntresetthe{fact}
\crefname{fact}{Fact}{Facts}

\newaliascnt{definition}{theorem}
\newtheorem{definition}[definition]{Definition}
\aliascntresetthe{definition}
\crefname{definition}{Definition}{Definitions}

\newaliascnt{remark}{theorem}

\aliascntresetthe{remark}
\crefname{remark}{Remark}{Remarks}

\newaliascnt{conjecture}{theorem}

\aliascntresetthe{conjecture}
\crefname{conjecture}{Conjecture}{Conjectures}

\newaliascnt{claim}{theorem}
\newtheorem{claim}[claim]{Claim}
\aliascntresetthe{claim}
\crefname{claim}{Claim}{Claims}

\newaliascnt{question}{theorem}

\aliascntresetthe{question}
\crefname{question}{Question}{Questions}

\newaliascnt{exercise}{theorem}

\aliascntresetthe{exercise}
\crefname{exercise}{Exercise}{Exercises}

\newaliascnt{example}{theorem}

\aliascntresetthe{example}
\crefname{example}{Example}{Examples}

\newaliascnt{notation}{theorem}

\aliascntresetthe{notation}
\crefname{notation}{Notation}{Notations}

\newaliascnt{problem}{theorem}

\aliascntresetthe{problem}
\crefname{problem}{Problem}{Problems}

\newcommand{\norm}[1]{\lVert#1\rVert}

\def\E{\mathbb E}
\newcommand{\Var}{{\bf Var}}

\newcommand{\R}{\mathbb R}

\newcommand{\dist}{\mathrm{d}}
\newcommand\fsnorm[1]{\ensuremath{\norm{#1}_F^2}} 
\newcommand{\1}{\mathbf1}
\newcommand\poly[1]{\ensuremath{\mathrm{poly}\left(#1\right)}}
\renewcommand{\P}{P_{\1}}

\title{Sample-Optimal Low-Rank Approximation of Distance Matrices}

\author{
  Piotr Indyk\thanks{\texttt{indyk@mit.edu}} \\
  MIT \\
  \and
  Ali Vakilian\thanks{\texttt{vakilian@mit.edu}} \\
  MIT \\
  \and
  Tal Wagner\thanks{\texttt{talw@mit.edu}} \\
  MIT \\
  \and
  David P. Woodruff\thanks{\texttt{dwoodruf@cs.cmu.edu}} \\
  Carnegie Mellon University \\
}

\begin{document}
\maketitle

\input{abstract}

\input{intro}

\input{prelim}

\input{algorithm}

\input{lower}

\input{experiment}

\paragraph{Acknowledgments.} P.~Indyk, A.~Vakilian and T.~Wagner were supported by funds from the MIT-IBM Watson AI Lab, NSF, and Simons Foundation. D.~Woodruff was supported partly by the National Science Foundation under Grant No. CCF-1815840, and this work was done partly while he was visiting the Simons Institute for the Theory of Computing. The authors would also like to thank Ainesh Bakshi for implementing the algorithm in this paper and producing our empirical results.

\bibliographystyle{amsalpha}
\bibliography{indyk19}

\appendix

\input{appendix_1}

\input{appendix_k}

\end{document}

%% file: abstract.tex
\begin{abstract}
A distance matrix $A \in \R^{n \times m}$ represents all pairwise distances, $A_{ij}=\dist(x_i,y_j)$, between two point sets $x_1,...,x_n$ and $y_1,...,y_m$ in an arbitrary metric space $(\mathcal Z, \dist)$. Such matrices arise in various computational contexts such as learning image manifolds, handwriting recognition, and multi-dimensional unfolding. 

In this work we study algorithms for low-rank approximation of distance matrices. Recent work by Bakshi and Woodruff (NeurIPS 2018) showed it is possible to compute a rank-$k$ approximation of a distance matrix in time $O((n+m)^{1+\gamma}) \cdot\poly{k,1/\epsilon}$, where $\epsilon>0$ is an error parameter and $\gamma>0$ is an arbitrarily small constant. Notably, their bound is sublinear in the matrix size, which is unachievable for general matrices.

We present an algorithm that is both simpler and more efficient. It reads only $O((n+m) k/\epsilon)$ entries of the input matrix, and has a running time of $O(n+m) \cdot \poly{k,1/\epsilon}$. We complement the sample complexity of our algorithm with a matching lower bound on the number of entries that must be read by any algorithm.
We provide experimental results to validate the approximation quality and running time of our algorithm.
\end{abstract}

%% file: intro.tex
\section{Introduction}

Computing low-rank approximations of matrices is a classic computational problem, with a remarkable number of applications in science and engineering. Given an $n\times m$ matrix $A$,   and a parameter $k$,  the goal is to compute a rank-$k$ matrix $A'$ that minimizes the approximation loss $\|A-A'\|_F$. Such an approximation can be found by computing the singular value decomposition of $A$. However, since the input matrix $A$ is often very large, faster approximate algorithms for computing low-rank approximations have been studied extensively (see the surveys \cite{mahoney2011randomized,woodruff2014sketching} and references therein). In particular, it is known for that for several interesting special classes of matrices, one can find an approximate low-rank solution using only a {\em sublinear} amount of time or samples from the input matrix. This  includes algorithms for {\em incoherent matrices}~\cite{candes2009exact}, {\em positive semidefinite matrices}~\cite{musco2017sublinear} and {\em distance matrices}~\cite{bakshi2018sublinear}. Note that sub-linear time or sampling bounds are not achievable for general matrices, as a single large entry in a matrix can significantly influence the output, and finding such an entry could take $\Omega(nm)$ time. 

In this paper we focus on computing low-rank approximations of distance matrices, i.e., matrices whose entries are induced by distances between points in some metric space. Formally: 

\begin{definition}[distance matrix]\label{def:distm}
A matrix $A\in\R^{n\times m}$ is called a~\emph{distance matrix} if there is an associated metric space $(\mathcal Z,\dist)$ with $\mathcal X=\{x_1,\ldots,x_n\}\subset\mathcal Z$ and $\mathcal Y=\{y_1,\ldots,y_m\}\subset\mathcal Z$, such that $A_{ij}=\dist(x_i,y_j)$ for every $i,j$.
\end{definition}

Distance matrices occur in many applications, such as learning image manifolds~\cite{weinberger2006unsupervised}, image understanding~\cite{tenenbaum2000global}, protein structure analysis~\cite{holm1993protein}, and more.
The recent survey~\cite{dokmanic2015euclidean} provides a comprehensive list.
Common software packages such as Julia, MATLAB or R include operations specifically design to produce or process such matrices. 

Motivated by these applications, a recent paper by~Bakshi and Woodruff~\cite{bakshi2018sublinear} introduced the first sub-linear time approximation algorithm for distance matrices. Their algorithm computes a rank-$k$ approximation of a distance matrix in time $O((n+m)^{1+\gamma}) \cdot\poly{k,1/\epsilon}$, where $\epsilon>0$ is an error parameter and $\gamma>0$ is an arbitrarily small constant. Specifically, it outputs matrices $V\in\R^{n\times k}$ and $U\in\R^{k\times m}$,
that satisfy an additive approximation guarantee of the form:

\[ \fsnorm{A-VU} \leq \fsnorm{A-A_k}+\epsilon\fsnorm{A} , \]

where $A_k$ is the optimal low-rank approximation of $A$.
The result of~\cite{bakshi2018sublinear} raises the question whether even faster algorithms for low-rank approximations of distance matrices are possible. This is the problem we address in this paper.

\subsection{Our results}

In this paper we present an algorithm for low-rank approximation of distance matrices that is both simpler and more efficient than the prior work. Specifically, we show:

\begin{theorem}[upper bound]\label{thm:ub}
There is a randomized algorithm that given a distance matrix $A\in\R^{n\times m}$, reads $O((n+m)k/\epsilon)$ entries of $A$, runs in time $\tilde O(n+m)\cdot\poly{k,1/\epsilon}$, and computes matrices $V\in\R^{n\times k},U\in\R^{k\times m}$ that with probability $0.99$ satisfy
\begin{equation}
\label{e:upper}
 \fsnorm{A-VU} \leq \fsnorm{A-A_k}+\epsilon\fsnorm{A} .
 \end{equation}
\end{theorem}

We complement the sample complexity of our algorithm with a matching lower bound on the number of entries of the input matrix that must be read by any algorithm.

\begin{theorem}[lower bound]\label{thm:lb}
Let $k\leq m\leq n$ and $\epsilon>0$ be such that $k/\epsilon=O(\min(m,n^{1/3}))$.
Any randomized and possibly adaptive algorithm that given a distance matrix $A\in\R^{n\times m}$, computes $V\in\R^{n\times k},U\in\R^{k\times m}$ that satisfy
$\fsnorm{A-VU} \leq \fsnorm{A-A_k}+\epsilon\fsnorm{A}$,
must read at least $\Omega((n+m)k/\epsilon)$ entries of $A$ in expectation.
The lower bound holds even for symmetric distance matrices.
\end{theorem}

We include an empirical evaluation of our algorithm on synthetic and real data.
The results validate that our approach attains good approximation with faster running time than existing methods.

\subsection{Our techniques}

\paragraph{Upper bound.} On a high level, our algorithm follows the approach of~\cite{bakshi2018sublinear}. The main idea is to use the result of Frieze, Kannan and Vempala~\cite{frieze2004fast}, which shows how to compute a solution satisfying Equation~\ref{e:upper} in $\tilde O(n+m)\cdot\poly{k,1/\epsilon}$ time, assuming the ability to sample a row (or a column) of the matrix $A$ with the probability at least proportional to its squared-norm.\footnote{Formally, the probability of selecting a given row $A_i$ should be $\Omega(\|A_i\|^2/\|A\|^2_F)$.} Thus the main challenge is to estimate row and column norms in sub-linear time. Although this cannot be done for general matrices, distance matrices have additional structure (imposed by  the triangle inequality), which makes the problem easier. Specifically, estimating column norms in distance matrices corresponds to computing, for all $x \in \mathcal{X}$,  the sum of distances (squared) from all points in $\mathcal{Y}$ to  $x$.  \cite{bakshi2018sublinear} gives a sampling algorithm that computes an $n^{0.9}$-approximation of those sums in sub-linear time. Since the approximation is pretty rough, they need to sample many more columns than the original algorithm of~\cite{frieze2004fast}, and then they apply the algorithm recursively to the sampled columns. The recursive nature of the algorithm makes the procedure and its analysis quite complex. 

To avoid this issue,  one could design an algorithm that computes a {\em constant-factor} approximation to row and column norms. In the symmetric case $\mathcal{X}=\mathcal{Y}$, this problem has been already studied in \cite{indyk1999sublinear}. Specifically, the latter paper developed an approximate comparator, which enables determining whether one row norm is approximately greater than another. Using standard sorting algorithms, one can approximately sort the rows  by  their norms using roughly $\tilde{O}(n)$ comparisons, where each comparison involves sampling roughly $\tilde{O}(1)$ entries of $A$.
Together with fully computing the norms of $\tilde{O}(1)$ landmark rows, this approximate sorting yields sufficiently approximations of all row norms.
Unfortunately, this approach does not immediately  generalize to the asymmetric case,  and exceeds the optimal number of samples by a few logarithmic factors. 

Our solution is to estimate row and column norms up to a constant factor with a {\em one-sided} guarantee. Specifically, we construct estimations that (a) do not {\em underestimate}  the true values and (b) the total sum of the estimations is comparable to the sum of the true values. This is sufficient to support a reduction to~\cite{frieze2004fast}, while making the estimation procedure much simpler. In the symmetric case, our procedure samples a random point $x^*$, and then estimates the sum of the distances from $x$ to $Y$ as a sum of the distance from $x$ to $x^*$ and the average distance from $x^*$ to $Y$. A simple application of the the triangle inequality shows this estimation provides the desired guarantee. We note that this idea was inspired by the construction of core-sets from~\cite{chen2009coresets}, although the technical development is quite different (and much simpler). 

After executing the algorithm of~\cite{frieze2004fast}, we still need one more step to compute the solution (as their method reports $U$ but not $V$). 
This amounts to a regression problem that can be solved by standard techniques.
To obtain a tight sampling bound that avoids any logarithmic factor in $k$, we use a recent solver of Chen and Price~\cite{chen2017condition}.

\paragraph{Lower bound.}
First let us note that an $\Omega(nk)$ lower bound can be easily obtained.
It is not hard to see that any $n\times k$ matrix with entries in $\{1,2\}$ is an (asymmetric) distance matrix, since the triangle inequality is satisfied trivially.
If we choose a uniformly random matrix from $\{1,2\}^{n\times k}$, then any algorithm that computes a matrix satisfying~\Cref{e:upper} with $\epsilon=\Omega(1)$, must match a $1-\Omega(\epsilon)$ fraction of the entries exactly, yielding the lower bound.

Our $\Omega(nk/\epsilon)$ lower bound is considerably more involved, and uses tools from communication complexity and random matrix theory.
For simplicity, let us describe our techniques in the case $k=1$.
Consider the problem of reporting the majority bit of a random binary string of length $r=\Theta(1/\epsilon)$.
This requires reading $\Omega(r)$ of the input bits.
If we stack together $n$ instances into an $n\times r$ random binary matrix $A$, then reporting the majority bit for a large fraction of rows requires reading $\Omega(nr)$ input bits.
This is our target lower bound.

The reduction proceeds by first shifting the values in $A$ from $\{0,1\}$ to $\{1,2\}$, so that it becomes an (asymmetric) distance matrix.
A na\"ive rank-$1$ approximation would be to replace each entry with $1.5$, yielding a total squared Frobenius error of $\tfrac14nr$.
However, the optimal rank-$1$ approximation is (essentially) to replace each row by its true mean value instead of $1.5$.
By anti-concentration of the binomial distribution, in most rows the majority bit appears $\Omega(\sqrt{r})$ times more often than the minority bit.
A simple calculation shows this leads to a constant additive advantage per row, and $\Omega(n)$ advantage over the whole matrix, of the optimal rank-$1$ approximation over the na\"ive one.
Since $\epsilon\fsnorm{A}=O(n)$, any algorithm that satisfies~\Cref{e:upper} must attain a similar advantage.

By spectral properties of random matrices, the optimal rank-$1$ approximation of $A$ is essentially unique.
In particular, the largest singular value of $A$ is much larger than the second-largest one.
For technical reasons we need to sharpen this separation even further.
We accomplish this by augmenting the matrix with an extra row with very large values, which corresponds to augmenting the metric space with an extra very far point.
As a result, any algorithm that satisfies~\Cref{e:upper} must approximately recover the mean values for a large fraction of the rows.
This allows us to solve the majority problem, by reporting whether each row mean in the rank-$1$ approximation matrix is smaller or larger than $1.5$.
This yields the desired lower bound for asymmetric distance matrices.
The result for symmetric distance matrices, and for general values of $k$, builds on similar techniques.



%% file: prelim.tex
\section{Preliminaries}

Consider a distance matrix $A\in\R^{n\times m}$ induced by two point sets, $x_1,...,x_n$ and $y_1,...,y_m$, as defined in~\Cref{def:distm}.
If $n=m$ and $x_i=y_i$ for every $i\in[n]$,\footnote{Throughout we use $[\ell]$ to denote $\{1,\ldots,\ell\}$ for an integer $\ell$.} then we call $A$ a~\emph{symmetric} distance matrix.
Otherwise we call it a~\emph{bipartite} distance matrix.
Throughout, $A_k$ denotes the optimal rank-$k$ approximation of $A$.

As mentioned earlier, our algorithm uses two sub-linear time algorithms as subroutines. They are formalized in the following two theorems.
The first reduces low-rank approximation to sampling proportionally to row (or column) norms.
We use $A_{i,*}$ to denote the $i$th row of $A$.
\begin{theorem}[\cite{frieze2004fast}]\label{thm:fkv}
Let $A\in\R^{n\times m}$ be any matrix. Let $S$ be a sample of $O(k/\epsilon)$ rows according to a probability distribution $(p_1,\ldots,p_n)$ that satisfies $p_i \geq \Omega(1)\cdot \norm{A_{i,*}}_2^2/\fsnorm{A}$ for every $i=1,\ldots,n$.
Then, in time $O(mk/\epsilon+\poly{k,1/\epsilon})$ we can compute from $S$ a matrix $U\in\R^{k\times m}$, that with probability $0.99$ satisfies
\begin{equation}\label{eq:fkv}
\fsnorm{A-AU^TU} \leq \fsnorm{A-A_k} + \epsilon\fsnorm{A} .
\end{equation}
\end{theorem}

The second result approximately solves a regression problem while reading only a small number of columns of the input matrix.
\begin{theorem}[\cite{chen2017condition}]\label{thm:cwcp}
There is a randomized algorithm that given matrices $A\in\R^{n\times m}$ and $U\in\R^{k\times m}$, reads only $O(k/\epsilon)$ columns of $A$, runs in time $O(mk)+\poly{k,1/\epsilon}$, and returns $V\in\R^{n\times k}$ that with probability $0.99$ satisfies
\begin{equation}\label{eq:cwcp}
\fsnorm{A-VU} \leq (1+\epsilon)\min_{X\in\R^{n\times k}}\fsnorm{A-XU} .
\end{equation}
\end{theorem}

Since our sampling procedure evaluates the sum of squared distances (rather than just distances), we need the following approximate version of the triangle inequality.
\begin{claim}\label{clm:triangle}
For every $x,y,z\in\mathcal Z$ in a metric space $(\mathcal Z,\dist)$, $\dist(x,y)^2\leq2(\dist(x,z)^2+\dist(z,y)^2)$.
\end{claim}
\begin{proof}
By the triangle inequality, $\dist(x,y)^2\leq(\dist(x,z) + \dist(z,y))^2=\dist(x,z)^2+2\dist(x,z)\dist(z,y)+\dist(z,y)^2$.
By the inequality of means, $\dist(x,z)\dist(z,y)\leq\tfrac12(\dist(x,z)^2+\dist(z,y)^2)$.
\end{proof}

%% file: algorithm.tex
\section{Algorithm}\label{sec:algorithm}

\newcommand{\INDENT}{\hspace{1em}}
\begin{algorithm}[!t]
\caption{Low-rank approximation for distance matrices}
\label{alg:main}
\smallskip
\textbf{Input:} Distance matrix $A\in\R^{n\times m}$. 
\textbf{Output:} Matrices $V\in\R^{n\times k}$ and $U\in\R^{k\times m}$.
\smallskip{\hrule height.2pt}\smallskip
\begin{algorithmic}[1]
   \STATE Choose $i^*\in[n]$ and $j^*\in[m]$ uniformly at random.
   \STATE For each $i=1,\ldots,n$: $p_i\leftarrow A_{i,j^*}^2 + A_{i^*,j^*}^2 + \frac1m\sum_{j=1}^mA_{i^*,j}^2$.
   \STATE Sample $O(k/\epsilon)$ rows of $A$ according to the distribution proportional to $(p_1,\ldots,p_n)$.
   \STATE Compute $U$ from the sample, using~\Cref{thm:fkv}.
   \STATE Compute $V$ from $A$ and $U$, using~\Cref{thm:cwcp}.
   \STATE Return $V,U$.
\end{algorithmic}
\end{algorithm}

In this section we prove~\Cref{thm:ub}.
The algorithm is stated in~\Cref{alg:main}.
The main step in the analysis is to provide guarantees for the sampling probabilities $p_i$ computed in Steps 1 and 2 of the algorithm. They are specified by the following theorem.

\begin{theorem}\label{thm:sampling}
There is a randomized algorithm that given a distance matrix $A\in\R^{n\times m}$, runs in time $O(m+n)$, reads $O(m+n)$ entries of $A$, and outputs sampling probabilities $(p_1,\ldots,p_n)$, that with probability $1-\delta$ satisfy $p_i \geq \Omega(\delta)\cdot \norm{A_{i,*}}_2^2/\fsnorm{A}$ for every $i=1,\ldots,n$.
\end{theorem}
\begin{proof}
Let $(\mathcal Z,\dist)$ be the metric space associated with $A$.
Let $\mathcal X=\{x_1,\ldots,x_n\}$ and $\mathcal Y=\{y_1,\ldots,y_m\}$ be the pointsets associated with its rows and its columns, respectively.
Choose a uniformly random $i^*\in[n]$ and a uniformly random $j^*\in[m]$.
For every $i\in[n]$, the output sampling probabilities are given by
\[ p_i = \dist(x_i,y_{j^*})^2 + \dist(x_{i^*},y_{j^*})^2 + \frac1m\sum_{j=1}^m\dist(x_{i^*},y_j)^2 . \]

All $p_i$'s can be computed in time $O(n+m)$ and by reading $n+m$ entries of $A$, since they only involve distances between $x_{i^*}$ to $\mathcal Y$ and between $y_{j^*}$ to $\mathcal X$. 
For every $i\in[n]$,

\begin{align*}
  \norm{A_{i,*}}_2^2 = \sum_{j=1}^m\dist(x_i,y_j)^2 &\leq 2\sum_{j=1}^m\left(\dist(x_i,y_{j^*})^2 + \dist(y_{j^*},y_j)^2\right) & \text{by~\Cref{clm:triangle}} \\
  &\leq 2\sum_{j=1}^m\left(\dist(x_i,y_{j^*})^2 + 2\dist(x_{i^*},y_{j^*})^2 + 2\dist(x_{i^*},y_j)^2\right) & \text{by~\Cref{clm:triangle}}\\
  &= 2m\cdot\dist(x_i,y_{j^*})^2 + 4m\cdot\dist(x_{i^*},y_{j^*})^2 + 4\sum_{j=1}^m\dist(x_{i^*},y_j)^2 &\\
  &\leq 4m\cdot p_i .&
\end{align*}

On the other hand, in expectation over $i^*$ and $j^*$ we have
$\E\left[\dist(x_i,y_{j^*})^2\right] = \frac1m\sum_{j=1}^m\dist(x_i,y_j)^2$, and 
$\E\left[\dist(x_{i^*},y_{j^*})^2\right] = \frac1{nm}\sum_{i=1}^n\sum_{j=1}^m\dist(x_i,y_j)^2$, and
$\E\left[\dist(x_{i^*},y_j)^2\right] = \frac1n\sum_{i=1}^n\dist(x_i,y_j)^2$.
Thus,

\begin{align*}
  \E\left[ \sum_{i=1}^np_i \right] & = \sum_{i=1}^n\left(\E\left[\dist(x_i,y_{j^*})^2\right] + \E\left[\dist(x_{i^*},y_{j^*})^2\right] +\E\left[\frac1m\sum_{j=1}^m\dist(x_{i^*},y_j)^2\right]\right) \\
  &= 3n\cdot \frac1{nm}\sum_{i=1}^n\sum_{j=1}^m\dist(x_i,y_j)^2 = \frac{3}{m}\fsnorm{A} .
\end{align*}

By Markov's inequality, $\sum_{i=1}^np_i\leq \frac{3}{\delta m}\fsnorm{A}$ with probability $1-\delta$.
Normalizing the $p_i$'s by their sum yields the theorem.
\end{proof}
We remark that if $A$ is a symmetric distance matrix, i.e., $\mathcal X=\mathcal Y$, the sampling probabilities can be simplified to choosing a single $i^*\in[n]$ uniformly at random, and letting $p_i = \dist(x_i,x_{i^*})^2 + \frac1n\sum_{j=1}^m\dist(x_{i^*},x_j)^2$. The proof is similar to the above.

\paragraph{Proof of~\Cref{thm:ub}.}
Consider~\Cref{alg:main}.
By~\Cref{thm:sampling}, the probabilities computed in Steps 1--2 are suitable for invoking~\Cref{thm:fkv}.
This ensures that the matrix $U$ computed in Steps 3--4 satisfies~\Cref{eq:fkv}.
\Cref{thm:cwcp} guarantees that the matrix $V$ computed in Step 5 satisfies~\Cref{eq:cwcp}.
Putting these together, we have
\begin{align*}
  \fsnorm{A-VU} & \leq (1+\epsilon)\min_{X\in\R^{n\times k}}\fsnorm{A-X^TU} & \text{by~\Cref{eq:cwcp}} \\
  &\leq (1+\epsilon)\fsnorm{A-AU^TU} & \\
  &\leq (1+\epsilon)\left(\fsnorm{A-A_k} + \epsilon\fsnorm{A}\right) & \text{by~\Cref{eq:fkv}} \\
  &\leq \fsnorm{A-A_k} + \epsilon\cdot(2+\epsilon)\cdot\fsnorm{A} & \text{since $\fsnorm{A-A_k}\leq\fsnorm{A}$,}
\end{align*}
and we can scale $\epsilon$ by a constant.
This proves~\Cref{e:upper}.
For the query complexity bound, observe that~\Cref{thm:fkv} reads $O(k/\epsilon)$ rows and~\Cref{thm:cwcp} reads $O(k/\epsilon)$ columns of the matrix, yielding a total of $O((n+m)k/\epsilon)$ queries. Finally, the running time is the sum of runnings times of~\Cref{thm:sampling,thm:fkv,thm:cwcp}.
\qed

%% file: lower.tex
\section{Lower Bound}\label{sec:lb}
For a clearer presentation, in this section we prove the lower bound in the special case $k=1$, for distance matrices that can be asymmetric (called~\emph{bipartite} in Definition~\ref{def:distm}).
This case encompasses the main ideas.
The full proof of~\Cref{thm:lb} appears in the appendix.
%
For concreteness, let us formally state the special case that will be proven in this section.
\begin{theorem}\label{thm:lb1}
Let $n,r,\epsilon$ be such that $r\leq n$ and $1>\epsilon\geq\Omega(n^{-1/3})$.
Any randomized algorithm that given a distance matrix $A\in\R^{n\times r}$, computes $V\in\R^{n\times k},U\in\R^{k\times r}$ that with probability $2/3$ satisfy
$\fsnorm{A-VU} \leq \fsnorm{A-A_1}+\epsilon\fsnorm{A}$,
must read at least $\Omega(n/\epsilon)$ entries of $A$ in expectation.
\end{theorem}

\subsection{Hard Distribution over Distance Matrices}\label{sec:distribution}
By Yao's principle, it suffices to construct a distribution over distance matrices, and prove the sampling lower bound for any deterministic algorithm that operates on inputs from that distribution.
We begin by defining a suitable distribution over distance matrices and proving some useful properties.

\paragraph{Hard problem.}
In the majority problem, the goal is to compute the majority bit of an input bitstring.
We will show the hardness of low-rank approximation via reduction from solving multiple random instances of the majority problem. The sample-complexity hardness of this problem is well-known, and is stated in the following lemma.
The proof is included in the appendix.
 
\begin{lemma}
\label{lmm:lb}
Let $r,t>0$ be integers.
Any deterministic algorithm that gets a uniformly random matrix $S\in\{0,1\}^{t\times r}$ as input, and outputs $s^*\in\{0,1\}^t$ such that for every $i\in[t]$,
$\Pr[s^*(i) = \text{majority element of $i$th row of $S$}] \geq 2/3$,
must read in expectation at least $\Omega(rt)$ entries of $S$.
\end{lemma}

\paragraph{The distribution.}
Given $n$ and $\epsilon>0$, let $\beta,C>0$ be constants that will be chosen later. ($\beta$ will be sufficiently small and $C$ sufficiently large.)
Let $r=\beta/\epsilon$, and assume w.l.o.g.~this is an integer by letting $\epsilon$ be sufficiently smaller.
Note that in Lemma~\ref{lmm:lb}, we can symbolically replace the majority alphabet $\{0,1\}$ with any alphabet of size $2$, and here we will use $\{1,2\}$.
Let $S\in\{1,2\}^{n\times r}$ be a uniformly random matrix.
Let $s_1,\ldots,s_n$ be its rows.
We call each of its rows an~\emph{instance} (of the majority problem).
Thus $S$ is an instance of the random multi-instance majority problem from Lemma~\ref{lmm:lb} (with $t=n$).
We begin by establishing some of its probabilistic properties.

Our goal is to solve $S$ via reduction to rank-$1$ approximation of distance matrices.
To obtain a distribution over distance matrices, we first take $S$ and randomly permute its rows to obtain a matrix $A$.
The random permutation is denoted by $\pi:[n]\rightarrow[n]$.\footnote{The random permutation is for a technical reason and does not change the distribution. Specifically, it is to prevent the algorithm in Lemma~\ref{lmm:lb} from focusing on a few fixed instances $\{s_i\}_{i=1}^{n'},$, $n'\ll n$, and never attempt the rest.}
Then, we add an additional $(n+1)$th row to $A$, whose entries are all equal $M=\sqrt{Cn}$.
The matrix with the added row is denoted by $\bar A$.

\paragraph{Metric properties.}
First we show that $\bar A$ is indeed a (bipartite) distance matrix.
\begin{lemma}\label{lmm:abarismetric}
Every supported $\bar A$ is a distance matrix.
\end{lemma}
\begin{proof}
Consider a symbolic pointset $X=P\cup Q$ where $P\cap Q=\emptyset$, such that $P=\{p_1,\ldots,p_{n+1}\}$ corresponds to the rows of $\bar A$, and $Q=\{q_1,\ldots,q_r\}$ to the columns of $\bar A$.
Our goal is to define a metric $\dist$ on $X$ such that $\dist(p_i,q_j)=\bar A_{ij}$ for every $i\in[n+1]$ and $j\in[r]$.
We need to set the rest of the distances such that $\dist$ is indeed a metric -- that is, such that $\dist$ satisfies the triangle inequality.
For every $i,i'\in[n]$ we set $\dist(p_i,p_{i'})=1$.
For every $j,j'\in[r]$ we set $\dist(q_j,q_{j'})=1$.
Finally we need to set the distances from $p_{n+1}$.
By construction of $\bar A$ we already have $\dist(p_{n+1},q_j)=M$ for every $j\in[r]$.
We set all the remaining distances, $\dist(p_i,p_{n+1})$ for every $i\in[n]$, to also be $M$.

We need to verify that for all distinct triplets $x,y,z\in X$, $\dist(x,y)\leq \dist(x,z)+\dist(z,y)$.
Indeed, all distances are in $\{1,2,M\}$. If $\dist(x,y)\in\{1,2\}$ then the inequality holds for any setting of $\dist(x,z)$ and $\dist(z,y)$. Otherwise $\dist(x,y)=M$, hence necessarily either $x=p_{n+1}$ or $y=p_{n+1}$, and in both cases $\dist(x,z)+\dist(z,y)\geq\max\{\dist(x,z),\dist(z,y)\}\geq \dist(p_{n+1},z)=M=\dist(x,y)$ as needed.
\end{proof}

\paragraph{Probabilistic properties.}
By anti-concentration of the binomial distribution, it is known that in a random length-$r$ bistring, the majority bit is likely appear $\Omega(\sqrt{r})$ times more than the other bit.

\begin{lemma}[anti-concentration]\label{lmm:antic}
Let $0<\delta<1$. Let $s\in\{1,2\}^r$ be a uniformly random majority instance.
Then, for $\gamma=\Omega(\delta)$, the majority element of $s$ appears in it at least $\tfrac12r+\gamma\sqrt{r}$ times with probability at least $1-\delta$.
\end{lemma}

We call an instance $s$ in $S$~\emph{typical} if its majority element appears in it at least $\frac12r+\gamma\sqrt{r}$ times, where $\gamma$ is the constant from Lemma~\ref{lmm:antic}.
Otherwise, we call the instance~\emph{atypical}.

Let $\Psi_{\text{typical}}$ denote the event there are at least $0.9n$ typical instances.
By Markov's inequality, $\Pr[\Psi_{\text{typical}}]\geq1-10\delta$.


%

\paragraph{Spectral properties.}
We will also require some facts from random matrix theory about the spectrum of $A$.
Let $\1$ denote the all-$1$'s vector in $\R^r$.
Let $\P$ denote the  orthogonal projection on the subspace spanned by $\1$.
The proofs of the following two lemmas are given in the appendix.

\begin{lemma}\label{lmm:rk1}
Suppose $r^3 = O(n)$.
With probability $1-e^{-\Omega(n/r^{3/2})}$,
$\fsnorm{A-A_1} \geq \tfrac14nr - O(n) $.
\end{lemma}

In the next lemma and throughout, $\norm{X}_2$ denotes the spectral norm of a matrix $X$.
\begin{lemma}\label{lmm:rk11}
Let $\zeta>0$.
Let $Z$ be a rank-$1$ matrix such that $\norm{Z}_2\leq\zeta\sqrt{n}$.
Then with probability $1-o(1)$,
$\fsnorm{A\P^\perp-Z} \geq \fsnorm{A\P^\perp} - \zeta\cdot O(n)$.
\end{lemma}

Since $r=\beta/\epsilon$ and in~\Cref{thm:lb1} we assume $\epsilon^{-3}=O(n)$, Lemma~\ref{lmm:rk1} is satisfied with probability $1-o(1)$. Therefore,
\begin{corollary}\label{cor:spectral}
Denote by $\Psi_{\text{spectral}}$ the event that the conclusions of both Lemmas~\ref{lmm:rk1}~and~\ref{lmm:rk11} hold.
Then $\Pr[\Psi_{\text{spectral}}]\geq1-o(1)$,
and therefore $\Pr[\Psi_{\text{typical}}\wedge\Psi_{\text{spectral}}]\geq1-10\delta-o(1)$.
\end{corollary}

\paragraph{Bounds on low-rank approximation}
Finally we give upper bounds on the approximation error allowed by~\Cref{e:upper}.
For every instance $s_i$, let $\mu_i=\frac{1}{r}\sum_{j=1}^rs_{ij}$ denote its mean.

\begin{lemma}\label{lmm:error}
$\fsnorm{\bar A-\bar A_1}+\epsilon\fsnorm{\bar A} \leq \sum_{i=1}^n\norm{s_i-\mu_i\1}_2^2 + (4+C)\beta n$.
\end{lemma}
\begin{proof}
Let $\bar A^*$ be the matrix in which each row equals $\1$ times the mean of the corresponding row of $A$. 
Then $\fsnorm{\bar A-\bar A_1}\leq\fsnorm{\bar A-\bar A^*} =  \sum_{i=1}^n\norm{s_i-\mu_i\1}_2^2$,
where the first inequality is since $\bar A^*$ has rank $1$ (each of its rows is a multiple of $\1$). 
Note that the sum ranges only up to $n$ and not $n+1$, since in the $(n+1)$th row all entries are equal (to $M$) and thus it contributes $0$ to $\fsnorm{\bar A-\bar A^*}$.
This bounds the first summand in the lemma.
To bound the second summand, note that each entry in the first $n$ rows of $\bar A$ is at most $2$, thus contributing in total $4rn$ to $\fsnorm{\bar A}$. The final row contributes $rM^2=Crn$.
Recalling that $\epsilon r=\beta$, we have $\epsilon\fsnorm{A}\leq (4+C)\beta n$.
\end{proof}

\begin{corollary}\label{cor:error}
$\fsnorm{\bar A-\bar A_1}+\epsilon\fsnorm{\bar A} \leq \tfrac14nr + (4+C)\beta n$.
\end{corollary}
\begin{proof}
For every $s_i$, its mean $\mu_i$ minimizes the sum of squared differences from a single value, namely
$\norm{s_i-\mu_i\1}_2^2=\min_{\nu\in\R}\norm{s_i-\nu\1}_2^2$.
In particular, $\norm{s_i-\mu_i\1}_2^2\leq\norm{s_i-1.5\cdot\1}_2^2$.
Furthermore, since $s_i\in\{1,2\}^r$, we have $\norm{s_i-1.5\cdot\1}_2^2=r\cdot(\tfrac12)^2=\tfrac14r$.
Hence $\sum_{i=1}^n\norm{s_i-\mu_i\1}_2^2\leq\tfrac14nr$, and the corollary follows from Lemma~\ref{lmm:error}.
\end{proof}

\subsection{Invoking the Algorithm}\label{sec:invoking}
Suppose we have a deterministic algorithm that given $\bar A$, returns $\bar A'=\bar a b^T$, where $\bar a\in\R^{n+1}$ and $b\in\R^r$, such that
\begin{equation}\label{eq:alg}
\fsnorm{\bar A-\bar ab^T}\leq\fsnorm{\bar A-\bar A_1}+\epsilon\fsnorm{\bar A} .
\end{equation}

Let $a\in\R^n$ denote the restriction of $\bar a$ to the first $n$ entries.
By scaling (i.e., multiplying $\bar a$ by a constant and $b$ by its reciprocal), we can assume w.l.o.g.~that $a_{n+1}=M$.
Since the $(n+1)$th row of $\bar A$ equals $M\cdot\1$, we have
\[ \fsnorm{\bar A-\bar ab^T} = \fsnorm{A-ab^T} + \norm{M\cdot\1-a_{n+1}b}_2^2 = \fsnorm{A-ab^T} + M^2\norm{\1-b}_2^2 . \]
If we rearrange this, and use~\Cref{eq:alg} and Corollary~\ref{cor:error} as an upper bound on $\fsnorm{\bar A-\bar ab^T}$ and Lemma~\ref{lmm:rk1} as a lower bound on $\fsnorm{A-ab^T}$, we get $M^2\norm{\1-b}_2^2 \leq (4+C)\beta n + O(n)$. Plugging $M=\sqrt{Cn}$,
\begin{equation}\label{eq:1b}
  \norm{\1-b}_2^2 \leq \left(\frac4C+1\right)\beta + \frac{O(1)}{C} = \frac{O(1)}{C} .
\end{equation}
This facts yields the following two lemmas, whose full proofs appear in the appendix.

\begin{lemma}\label{lmm:bproj}
We have $1-\eta/\sqrt{r} \leq \norm{b^T\P}/\norm{\1} \leq 1+\eta/\sqrt{r}$, 
where $\eta>0$ is a constant that can be made arbitrarily small by choosing $C>0$ sufficiently large.
\end{lemma}

\begin{lemma}\label{lmm:a}
$\norm{a}=O(\sqrt{n})$.
\end{lemma}

\subsection{Solving Majority}
We now show how to use $\bar A'=\bar ab^T$ to solve the majority instance $S$ of the problem in Lemma~\ref{lmm:lb}.
We condition on the intersection of the events $\Psi_{\text{typical}}$ and $\Psi_{\text{spectral}}$.
By Corollary~\ref{cor:spectral} it occurs with probability at least $1-10\delta-o(1)$.

Let $s_1,\ldots,s_n\in\{1,2\}^r$ denote the random instances in $S$.
Recall that we assigned them to rows of $A$ by a uniformly random permutation $\pi$, that is, the $\pi(i)$th row of $A$ equals $s_i$.

We use $a$ to solve the majority problem as follows.
For each $s_i$, if $a_{\pi(i)}\leq1.5$ then we output that the majority is $1$, and otherwise we output that the majority is $2$.
We say that $a$~\emph{solves} the instance $s_i$ if the output is correct.
Due to $\pi$ being random, the probability that $A$ solves any instance $s_i$ is identical.
Denote this probability by $p$. We need to show that $p\geq2/3$.

Assume by contradiction that $p<2/3$.
By Markov's inequality, with probability at least $4/5$ we have at least $n/6$ unsolved instances.
Since by $\Psi_{\text{typical}}$ there are only $0.1n$ atypical instances, we have at least $n/15$ unsolved typical instances.
Denote by $S'$ the set of unsolved typical instances.
Consider such instance $s_i\in S'$. Suppose its majority element is $1$. Then, since it is typical,
  \begin{align}
    \norm{s_i-\mu_i\1}_2^2 &\leq \norm{s_i-(1.5-\gamma/\sqrt{r})\1}_2^2 \nonumber \\
    &\leq \left(\tfrac12r+\gamma\sqrt{r}\right)\left(\tfrac12-\gamma/\sqrt{r}\right)^2 + \left(\tfrac12r-\gamma\sqrt{r}\right)\left(\tfrac12+\gamma/\sqrt{r}\right)^2 \nonumber \\
    &= \tfrac14r-\gamma^2. \label{eq:typical}
  \end{align}
On the other hand, since $s_i$ is unsolved then $a_{\pi(i)}\geq1.5$.
Hence by Lemma~\ref{lmm:bproj}, $\frac{\norm{b^T\P}_2}{\norm{\1}_2}\cdot a_{\pi(i)}\geq1.5-\eta/\sqrt{r}$. Therefore, noting that $b^T\P = \frac{\norm{b^T\P}_2}{\norm{\1}_2}\cdot\1$, we have
\begin{align}
  \norm{s_i - a_{\pi(i)}b^T\P}_2^2 &= \norm{s_i - \frac{\norm{b^T\P}_2}{\norm{\1}_2}\cdot a_{\pi(i)}\1}_2^2 \nonumber \\
  & \geq \norm{s_i - (1.5-\tfrac{\eta}{\sqrt{r}})\cdot\1}_2^2 \nonumber \\
  &\geq \left(\tfrac12r+\gamma\sqrt{r}\right)\left(\tfrac12-\eta/\sqrt{r}\right)^2 + \left(\tfrac12r-\gamma\sqrt{r}\right)\left(\tfrac12+\eta/\sqrt{r}\right)^2 \nonumber \\
  &= \tfrac14r + \eta^2 - 2\gamma\eta . \label{eq:unsolved}
\end{align}
Similar calculations yield the same bounds when the majority element is $2$.
From~\Cref{eq:typical}, together with~\Cref{eq:alg} and Lemma~\ref{lmm:error}, we get:
\begin{equation}\label{eq:ub}
  \fsnorm{\bar A-\bar ab^T} \leq \sum_{i=1}^n\norm{s_i-\mu_i\1}_2^2 \leq \tfrac{1}{15}n(\tfrac14r-\gamma^2) + \sum_{s_i\notin S'}\norm{s_i-\mu_i\1}_2^2 + (4+C)\beta n . 
\end{equation}
On the other hand, by~\Cref{eq:unsolved},
\begin{equation}\label{eq:lb}
  \fsnorm{A-ab^T\P} = \sum_{i=1}^n\norm{s_i-a_{\pi(i)}b^T\P}_2^2 \geq \tfrac{1}{15}n(\tfrac14r + \eta^2 - 2\gamma\eta) + \sum_{s_i\notin S'}\norm{s_i-\mu_i\1}_2^2 . 
\end{equation}
It remains to relate~\Cref{eq:ub,eq:lb} to derive a contradiction.
By the Pythagorean identity,
$\fsnorm{A - a b^T} = \fsnorm{A\P - a b^T\P} + \fsnorm{A\P^\perp - a b^T\P^\perp}$,
and
\[
  \fsnorm{A - a b^T\P} = \fsnorm{A\P - a b^T\P^2} + \fsnorm{A\P^\perp - a b^T\P\P^\perp} = \fsnorm{A\P - a b^T\P} + \fsnorm{A\P^\perp} .
\]
Together,
$\fsnorm{A - a b^T\P} = \fsnorm{A - a b^T} + \left(\fsnorm{A\P^\perp} - \fsnorm{A\P^\perp - a b^T\P^\perp}\right)$.
Let us upper-bound both terms.
For the first term, we simply use $\fsnorm{A - a b^T}\leq\fsnorm{\bar A - \bar a b^T}$.
For the second term, note that
$\norm{b^T\P^\perp}_2 = \norm{\1\P^\perp - b^T\P^\perp}_2 \leq \norm{\1- b}_2$.
Together with Lemma~\ref{lmm:a},
$\norm{ab^T\P^\perp}_2\leq O(\sqrt n)\cdot\norm{1-b}_2$.
Thus by Lemma~\ref{lmm:rk11},
$\left(\fsnorm{A\P^\perp} - \fsnorm{A\P^\perp - a b^T\P^\perp}\right)\leq O(n)\cdot\norm{\1-b}_2$.
By~\Cref{eq:1b}, the latter is $O(n)/\sqrt{C}$.
Plugging both upper bounds,
\[
  \fsnorm{A - a b^T\P} \leq \fsnorm{\bar A - \bar a b^T} + O(n)\cdot C^{-1/2}.
\]
This relates~\Cref{eq:ub,eq:lb}, yielding
\[
  \tfrac{1}{15}(\gamma^2 + \eta^2) \leq \tfrac{2}{15}\cdot \gamma\eta + (4+C)\beta + O(1)\cdot C^{-1/2}.
\]
Since $\gamma$ is fixed, choosing $\beta,\eta$ sufficiently small and $C$ sufficiently large leads to a contradiction.

Thus $p\geq2/3$, meaning the reduction solves each instance in the majority problem $S$ with probability at least $2/3$.
Accounting for the conditioning on $\Psi_{\text{typical}}$ and $\Psi_{\text{spectral}}$, the determined low-rank approximation algorithm from~\Cref{sec:invoking} solves a random instance of Lemma~\ref{lmm:lb} with probability at least $2/3-10\delta-o(1)$ (the constants can be scaled without changing the lower bound).
Hence, it requires reading at least $\Omega(n/\epsilon)$ bits from the matrix, which proves~\Cref{thm:lb1}.

%% file: experiment.tex
\section{Experiments}

In this section, we evaluate the empirical performance of~\Cref{alg:main} compared to the existing methods in the literature: 
conventional SVD, the algorithm of~\cite{bakshi2018sublinear} (BW), and the input-sparsity time algorithm of Clarkson and Woodruff~\cite{clarkson2017low} (IS).   
For SVD we use numpy's linear algebra package.\footnote{\href{https://docs.scipy.org/doc/numpy-1.15.1/reference/routines.linalg.html}{https://docs.scipy.org/doc/numpy-1.15.1/reference/routines.linalg.html}. This performs full SVD. The iterative SVD algorithms built into MATLAB and Python yielded errors larger by a few orders of magnitude than the reported methods, so they are not included.} 
%
The experimental setup is analogous to that in~\cite{bakshi2018sublinear}. Specifically,  we consider two datasets:
\begin{itemize}
\item{\bf Synthetic clustering dataset.} This data set is generated using the  {\tt scikit-learn} 
package. We generate $10,000$ points with $200$ features and partition the points into $20$ clusters. As observed in our experiments, the dataset is expected to have a good rank-$20$ approximation.
\item{\bf MNIST dataset.} The dataset contains $70,000$ handwritten characters, and each is considered a point. We subsample $10,000$ points.
\end{itemize}
For each dataset we construct a symmetric distance matrix $A_{i,j} = \dist(p_i, p_j)$.
We use four distances $\dist$: Manhattan ($\ell_1$), Euclidean ($\ell_2$), Chebyshev ($\ell_{\infty}$) and Canberra\footnote{The Canberra distance $\dist_c$ between vectors $p, q \in \mathbb{R}^n$ is defined as $\dist_c(p,q) = \sum_{i=1}^n {|p_i - q_i| \over |p_i| + |q_i|}$.} ($\ell_c$).
Figures~\ref{fig:synthetic} and~\ref{fig:mnist} show the approximation error for each distance on each dataset, for varying values of the rank $k$. Note that SVD achieves the optimal approximation error.
Table~\ref{tab:synthetic} lists the running times for $k=40$.
Figure~\ref{fig:asymptotic} shows the  running time of our algorithm for MNIST subsampled to varying sizes, for $k=40$.


\begin{figure}[htbp]
\centering
	\caption{The approximation error of the four algorithms on the synthetic clustering dataset.}
	\label{fig:synthetic}
	{%
		\begin{minipage}{.48\textwidth}%
			\includegraphics[width=\textwidth]{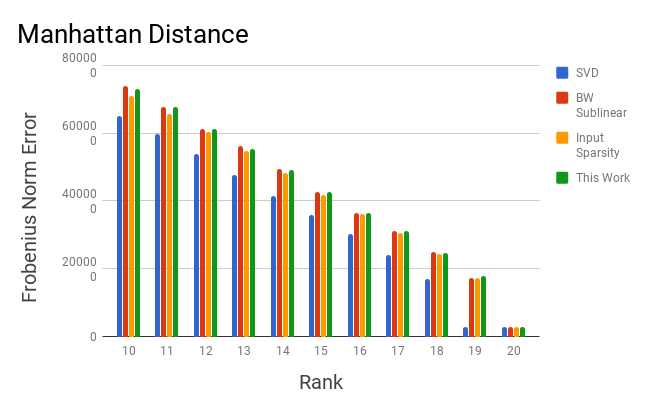}
		\end{minipage}\quad 
		\begin{minipage}{.48\textwidth}%
			\includegraphics[width=\textwidth]{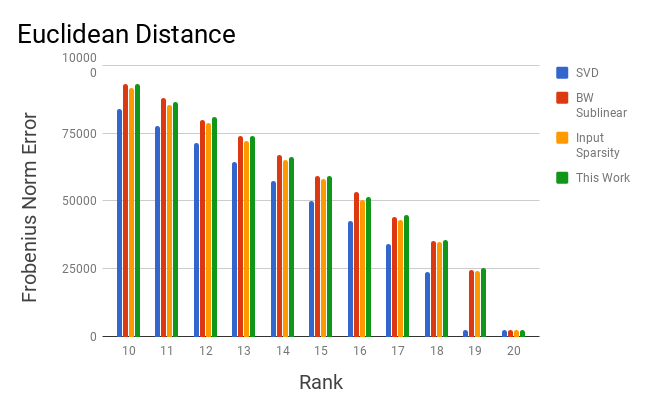}
		\end{minipage}\quad 
		\begin{minipage}{.48\textwidth}%
			\includegraphics[width=\textwidth]{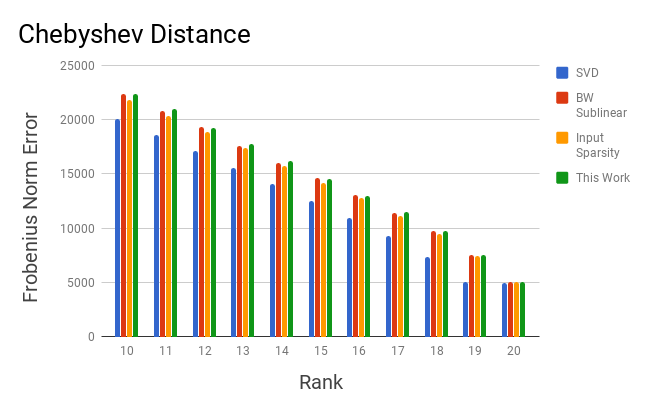}
		\end{minipage}\quad 
		\begin{minipage}{.48\textwidth}%
			\includegraphics[width=\textwidth]{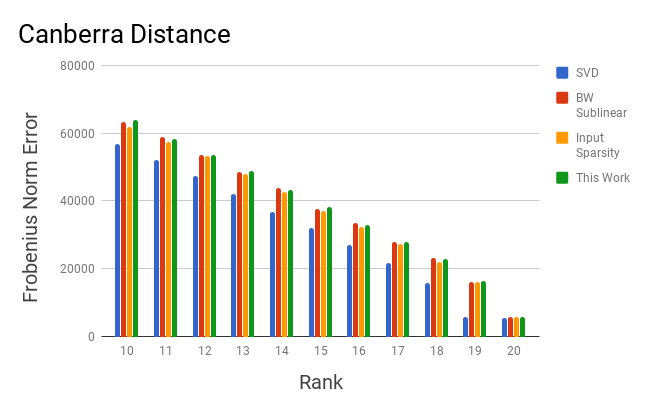}
		\end{minipage}\quad 
	}
\end{figure}

\begin{figure}[htbp]
\centering
	\caption{The approximation error of the four algorithms on the MNIST dataset.}
	\label{fig:mnist}
	{%
		\begin{minipage}{.48\textwidth}%
			\includegraphics[width=\textwidth]{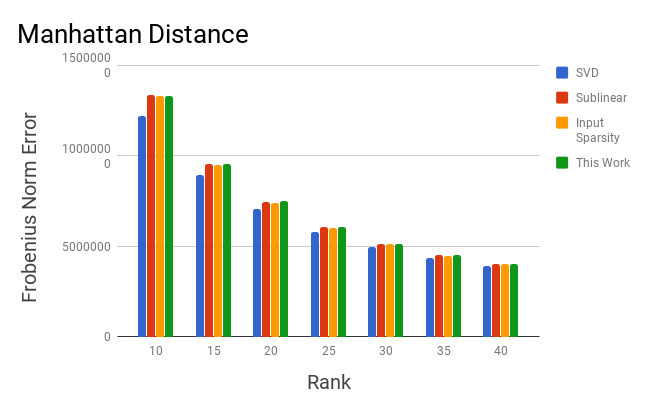}
		\end{minipage}\quad 
		\begin{minipage}{.48\textwidth}%
			\includegraphics[width=\textwidth]{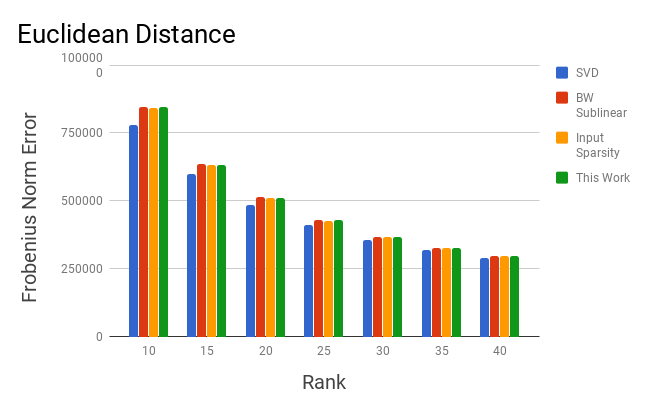}
		\end{minipage}\quad 
		\begin{minipage}{.48\textwidth}%
			\includegraphics[width=\textwidth]{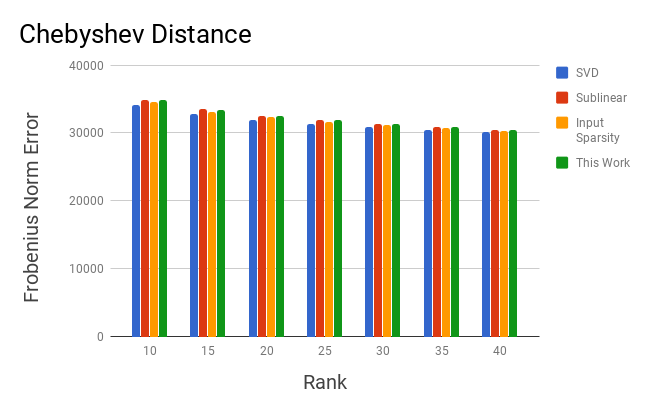}
		\end{minipage}\quad 
		\begin{minipage}{.48\textwidth}%
			\includegraphics[width=\textwidth]{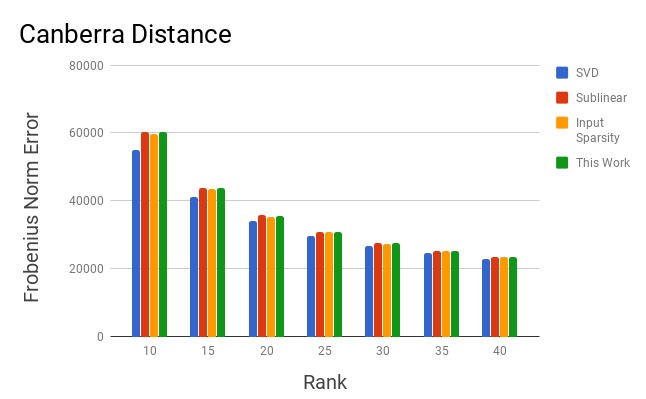}
		\end{minipage}\quad 
	}
\end{figure}

\begin{table}[htbp]
\centering
	\caption{Running times (in seconds) of the compared methods for rank $k=40$ approximation}
	\label{tab:synthetic}
	{%
		\begin{tabular}{|c||cccc||cccc|}
		    \hline
		    & \multicolumn{4}{|c||}{\bfseries Synthetic} & \multicolumn{4}{|c|}{\bfseries MNIST} \\
			\hline
			\bfseries Metric & \bfseries SVD & \bfseries IS & \bfseries BW & \bfseries This Work & \bfseries SVD & \bfseries IS & \bfseries BW & \bfseries This Work\\
			\hline
			$\ell_2$ & 398.77 &	8.95 & 1.70 & 1.17& 398.50 & 34.32 & 4.17 & 1.23\\
			$\ell_1$ & 410.60 & 8.16 & 1.82 & 1.197 & 560.91 & 39.50 & 3.71 & 1.23\\
			$\ell_\infty$ & 427.90 & 9.18 & 1.63 & 1.16 & 418.01 & 39.33 & 4.00 & 1.14\\
			$\ell_c$ & 452.17 & 8.49 & 1.76 & 1.15 & 390.07 & 38.34 & 3.91 & 1.24 \\
			\hline
		\end{tabular}
	}
\end{table}

\begin{figure}[htbp]
\centering
	\caption{Running time of our algorithm on subsets of MNIST, for $k=40$.}
	\label{fig:asymptotic}
	{\includegraphics[width=0.44\textwidth]{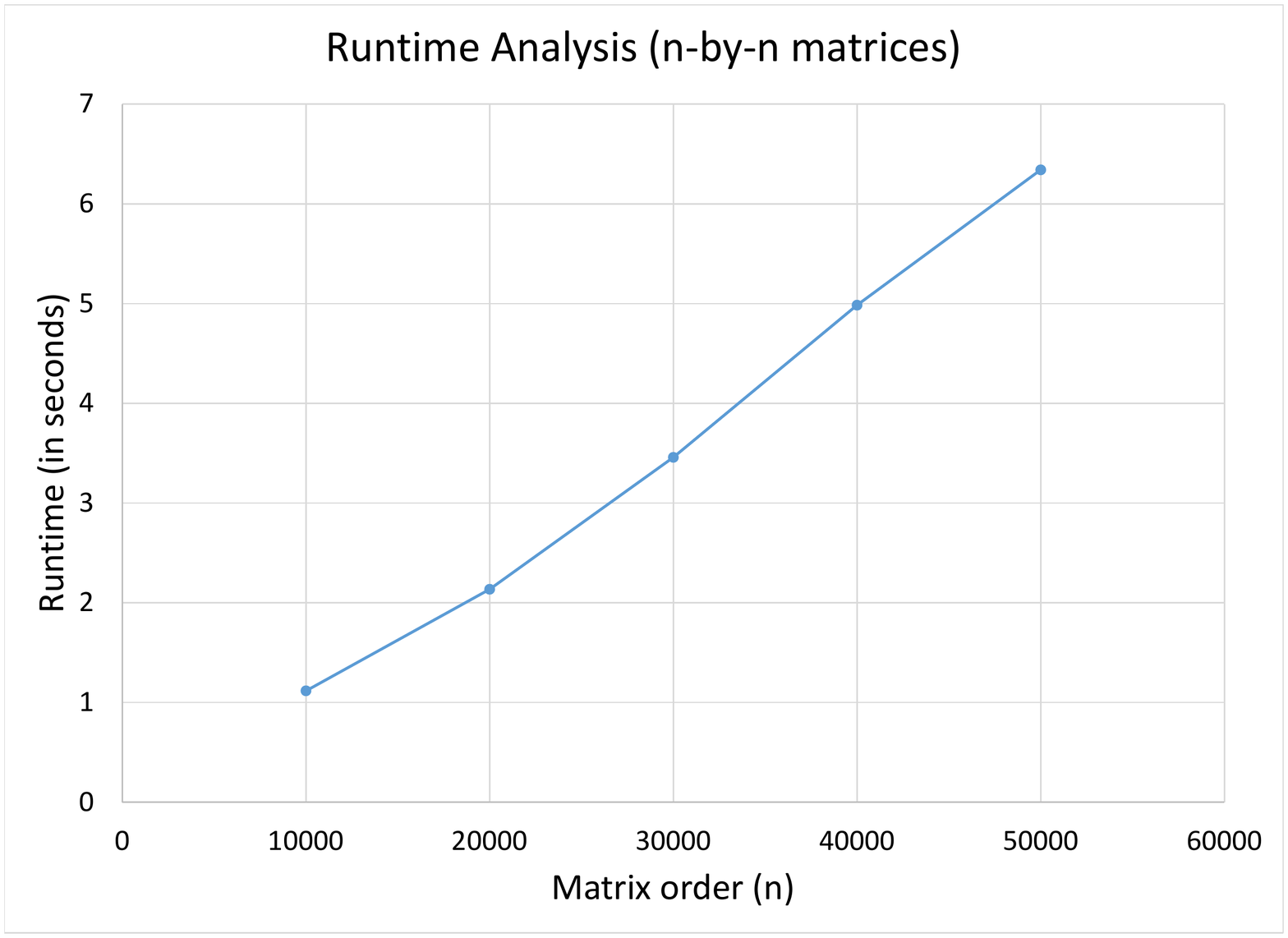}}
\end{figure}



%% file: appendix_1.tex
\section{Deferred Proofs from~\Cref{sec:lb}}

\subsection{Preliminary Lemmas}
The following lemmas are known and we include their proofs for completeness.

\begin{lemma}[hardness of majority, Lemma~\ref{lmm:lb} restated]
Let $r,t>0$ be integers.
Any deterministic algorithm that gets a uniformly random matrix $S\in\{0,1\}^{t\times r}$ as input, and outputs $s^*\in\{0,1\}^t$ such that for every $i\in[t]$,
$\Pr[s^*(i) = \text{majority element of $i$th row of $S$}] \geq 2/3$,
must read in expectation at least $\Omega(rt)$ entries of $S$.
\end{lemma}
\begin{proof}
The reduction is from the distributional Gap-Hamming communication problem, which is defined as follows.
Alice has a bit string $x$ in $\{0,1\}^r$ and Bob has a bit string $y$ in $\{0,1\}^r$, where $x$ and $y$ are independent and uniformly distributed. Let $\Delta(x,y)$ denote their Hamming distance.
The goal is to decide whether $\Delta(x,y) \geq \tfrac12r + \sqrt{r}$ or $\Delta(x,y) \leq \tfrac12r - \sqrt{r}$.
If neither case holds, then any output is considered successful.
The information cost under this distribution is $\Omega(r)$~\cite{braverman2016information}.

Next consider the $r$-fold version of the same problem, i.e., Alice and Bob are given $r$ instances of distributional Gap-Hamming, and they need to solve a constant fraction of them.
By a standard direct sum theorem (see e.g.~\cite{braverman2014information}), this requires $\Omega(rt)$ bits of communication.

Finally we reduce this problem to the majority problem in lemma statement.
Let $X$ denote the xor matrix of Alice's and Bob's matrices.
The Gap Hamming problem is equivalent to finding the majority bit over rows of $X$ in which the majority bit appears at least $\tfrac12r + \sqrt{r}$ times (call these rows ``typical'').
By Lemma~\ref{lmm:antic} this happens in a large constant fraction of the rows.
Given a black-box algorithm for the majority problem that queries $q$ entries of the input matrix, Alice and Bob can simulate it on $M$ by communicating to each other only those entries of their matrices, which costs them $\Theta(q)$. The algorithm solves a large fraction of the rows, and thus a large fraction of the typical rows.
Hence they have solved the Gap Hamming problem, and $q=\Omega(rt)$.
\end{proof}

\begin{lemma}[binomial anti-concentration, Lemma~\ref{lmm:antic} restated]
Let $0<\delta<1$. Let $s\in\{1,2\}^r$ be a uniformly random majority instance.
Then, for $\gamma=\Omega(\delta)$, the majority element of $s$ appears in it at least $\tfrac12r+\gamma\sqrt{r}$ times with probability at least $1-\delta$.
\end{lemma}
\begin{proof}
Let $X\sim\mathrm{Binomial}(r,\tfrac12)$.
The statement we need to show is equivalent to $\Pr[|X-\tfrac12r|<\tfrac12\gamma\sqrt{r}]<\delta$, or equivalently,\[ \sum_{i=\lfloor\tfrac12r-\tfrac12\gamma\sqrt{r}\rfloor}^{\lceil\tfrac12r+\tfrac12\gamma\sqrt{r}\rceil}\Pr[X=i] < \delta . \]
Let $\gamma=\sqrt{\pi/2}\cdot\delta$.
Note that $\Pr[X=i]\leq\Pr[X=\lfloor r/2 \rfloor]$ for every $i$, and $\Pr[X=\lfloor r/2 \rfloor]=2^{-r}{r\choose\lfloor r/2 \rfloor}\leq1/\sqrt{2\pi r}$ by a known estimate. Therefore, the above left-hand side sum is upper-bounded by $2\gamma\sqrt{r}/\sqrt{2\pi r}=\delta$ as needed.
\end{proof}

\subsection{Spectral Properties}

For the next two lemmas, write $A=1.5J+B$ where $J$ is the all-$1$'s matrix and $B$ is a matrix with i.i.d.~random entries chosen uniformly from $\{-\tfrac12,\tfrac12\}$.

\begin{lemma}[Lemma~\ref{lmm:rk1}, restated]
Suppose $r^3 = O(n)$.
With probability $1-e^{-\Omega(n/r^{3/2})}$,
$\fsnorm{A-A_1} \geq \tfrac14nr - O(n) $.
\end{lemma}
\begin{proof}
We use a sharp estimate of~\cite{feldheim2010universality} on the smallest singular value of $B$ (see also~\cite{rudelson2010non}, eq.~(2.5)).
It states that with probability $1-\exp(-\Omega(n/r^{3/2}))$, all $r$ singular values of $B$ are at least $\tfrac14n-O(\sqrt{nr})$.
Furthermore, since $A$ is obtained from $B$ by adding a rank-$1$ matrix ($1.5J$), then by Theorem 1 of~\cite{thompson1976behavior}, $A$ has at least $r-1$ singular values which are at least $\tfrac14n-O(\sqrt{nr})$.
Therefore $\fsnorm{A-A_1}$, which is the sum of all squared singular values of $A$ except the largest, is at least
$(r-2)\cdot(\tfrac14n-O(\sqrt{nr})) = \tfrac14nr - O(n) - O(r^{3/2}\sqrt{n})$.
\end{proof}

\begin{lemma}[Lemma~\ref{lmm:rk11}, restated]
Let $\zeta>0$.
Let $Z$ be a rank-$1$ matrix such that $\norm{Z}_2\leq\zeta\sqrt{n}$.
Then with probability $1-o(1)$,
$\fsnorm{A\P^\perp-Z} \geq \fsnorm{A\P^\perp} - \zeta\cdot O(n)$.
\end{lemma}
\begin{proof}
By the Hoffman-Wielandt inequality~\cite{hoffman2003variation} for singular values\footnote{See also Exercise 22(v) in~\cite{taoblog}}, any $X,Y\in\R^{n\times r}$ satisfy $
  \fsnorm{X-Y} \geq \sum_{j=1}^r(\sigma_j(X)-\sigma_j(Y))^2$,
where $\sigma_1(X)\geq\ldots\geq\sigma_r(X)$ and $\sigma_1(Y)\geq\ldots\geq\sigma_r(Y)$ are their respective sorted singular values.
In particular, letting $X=A\P^\perp$ and $Y=Z$, since $Z$ has rank $1$, we have
\begin{align*}
\fsnorm{A\P^\perp-Z} &\geq \sum_{j=2}^r(\sigma_j(A\P^\perp))^2 + \left(\sigma_1(A\P^\perp)-\sigma_1(Z)\right)^2 \\
&= \fsnorm{A\P^\perp} - 2\sigma_1(A\P^\perp)\sigma_1(Z) + (\sigma_1(Z))^2 \\
&= \fsnorm{A\P^\perp} - 2\norm{A\P^\perp}_2\norm{Z}_2 + \norm{Z}_2^2 .
\end{align*}
Therefore it suffices to show that $\norm{A\P^\perp}_2 = O(\sqrt n)$.
Indeed,
$\norm{A\P^\perp}_2 = \norm{(1.5J+B)\P^\perp}_2 = \norm{B\P^\perp}_2 \leq \norm{B}_2$,
and an upper bound $\norm{B}_2=O(\sqrt n)$ is well known, e.g., see~Proposition 2.4 in \cite{rudelson2010non}.\end{proof}

\subsection{Lemmas from~\Cref{sec:invoking}}
\begin{lemma}[Lemma~\ref{lmm:bproj}, restated]
We have
\[ 1-\frac{\eta}{\sqrt{r}} \leq \frac{\norm{b^T\P}}{\norm{\1}} \leq 1+\frac{\eta}{\sqrt{r}} , \] 
where $\eta>0$ is a constant that can be made arbitrarily small by choosing $C>0$ sufficiently large.
\end{lemma}
\begin{proof}
By the triangle inequality we have
\[ \norm{\1}_2-\norm{\1-b}_2\leq\norm{b}_2\leq\norm{\1}_2+\norm{\1-b}_2 .\]
The upper bound implies
\begin{equation}\label{eq:triupper}
\norm{b^T\P}\leq\norm{b}\leq\norm{\1}_2+\norm{\1-b}_2 .
\end{equation}
The lower bound implies
\[
  \norm{\1-b}_2^2 = \norm{\1}_2^2 + \norm{b}_2^2 - 2b^T\1 \geq \norm{\1}_2^2 + \norm{\1}_2^2 + \norm{\1-b}_2^2 -2\norm{\1}_2\norm{\1-b}_2 - 2b^T\1,
\]
which rearranges to $b^T\1\geq\norm{\1}_2^2-\norm{\1}_2\norm{\1-b}_2$, implying 
\begin{equation}\label{eq:trilower}
\norm{b^T\P}=b^T(\frac{1}{\norm{\1}_2}\1)\geq\norm{\1}_2-\norm{\1-b}_2 .
\end{equation}
Putting~\Cref{eq:triupper,eq:trilower} together,
\[ \norm{\1}_2 - \norm{\1-b}_2 \leq \norm{b^T\P} \leq \norm{\1}_2 + \norm{\1-b}_2 ,\]
and the lemma follows since $\norm{\1}_2=\sqrt{r}$ and since by~\cref{eq:1b}, $\norm{\1-b}_2$ is a constant that can be made arbitrarily small by choosing $C>0$ sufficiently large.
\end{proof}

\begin{lemma}[Lemma~\ref{lmm:a}, restated]
$\norm{a}=O(\sqrt{n})$.
\end{lemma}
\begin{proof}
By the triangle inequality, $\norm{b}_2 \geq \norm{\1}_2-\norm{b-\1}_2$.
Since $\norm{\1}_2=\sqrt r$ and $\norm{b-\1}_2$ is an arbitrarily small constant by~\Cref{eq:1b}, $\norm{b}_2 \geq \tfrac12\sqrt r$.
Thus
\begin{equation}\label{eq:aaux}
  \fsnorm{ab^T}=\norm{a}_2^2\norm{b}_2^2\geq\tfrac14r\norm{a}_2^2 .
\end{equation}
We finish by showing that $\fsnorm{ab^T}=O(nr)$. Indeed,
\begin{align*}
\fsnorm{A-ab^T} &\leq \fsnorm{\bar A-\bar ab^T} &  \\
&\leq \fsnorm{\bar A-\bar A_1}+\epsilon\fsnorm{\bar A} & \text{by~\Cref{eq:alg}} \\
&\leq \tfrac14nr + (4+C)\beta n . & \text{by Corollary~\ref{cor:error}}
\end{align*}
Furthermore $\fsnorm{A}\leq4nr$ since each entry of $A$ has absolute value at most $2$.
Finally, by approximate triangle inequality (\Cref{clm:triangle}),
\[ 
 \fsnorm{ab^T} \leq 2\fsnorm{A} + 2\fsnorm{A-ab^T} = O(nr) .
\]
With~\Cref{eq:aaux} this implies the lemma.
\end{proof}

\subsection{Lower Bound for Symmetric Distance Matrices}\label{sec:symmetriclb}

In this section we show that the lower bound in~\Cref{thm:lb1} applies also to symmetric distance matrices.
The proof is by a reduction to the asymmetric case.

\subsubsection{General $k$}
We start by reducing rank-$k$ approximation of asymmetric distance matrices to rank-$(2k+2)$ approximation of symmetric distance matrices.
By tuning $\beta$, suppose w.l.o.g.~that $kr=\beta k/\epsilon$ is a divisor of $n$.
Let $B\in\R^{n\times kr}$ be an asymmetric distance matrix drawn from the hard distribution. 
Recall that all entries are in $\{1,2,3\}$. We can scale them by half so they are all the interval $[1,2]$.

We construct a symmetric distance matrix $A\in\R^{2n\times 2n}$. It is partitioned into $n\times n$ blocks,
$$A = \left[
\begin{array}{cc}
A_{11} & A_{12} \\
A_{21} & A_{22}
\end{array} 
\right] .$$
We set its entries as follows.
Its main diagonal is all-zeros.
$A_{21}$ consists of $n/(kr)$ copies of $B$, concatenated horizontally.
$A_{11}$ has all off-diagonal entries set to $1$.
$A_{12}$ and $A_{22}$ are determined symmetrically.
Since all entries are in the interval $[1,2]$, the triangle inequality is satisfied trivially and thus $A$ is a distance matrix.

We will show here that any rank-$(2k+2)$ approximation algorithm for $A$ must read at least $\Omega(nk/\epsilon)$ of its entries.
Let $\mathbf0$ and $\1$ denote the all-$0$'s and all-$1$'s vectors in $\R^n$, respectively.
For any $x,y\in\R^n$ let $[x\;y]$ denote their concatenation into a vector in $\R^{2n}$.
Let $B_k$ be the optimal rank-$k$ approximation of $B$.
Write $B$ as $B=UV^T$ where $U\in\R^{n\times k}$ and $V\in\R^{kr\times k}$.
Let $u_1,\ldots,u_k$ be the columns of $U$ and let $v_1,\ldots,v_k$ be the columns of $V^T$.
For every $i\in[k]$, let $\bar v_i$ be the vector given by concatenating $n/(kr)$ copies of $v_i$.
Consider the rank-$(2k+2)$ approximation of $A$ given by the column vectors $\{[\mathbf0\; u_i]:i\in[k]\} \cup \{[\bar v_i\;\mathbf0]:i\in[k]\} \cup \{[\1\;\mathbf0],[\mathbf0\;\1]\}$.
Let us bound its error in approximating $A$.
On $A_{12}$ and $A_{21}$, which contain concatenated copies of $B$, the vectors $\{[\mathbf0\; u_i]:i\in[k]\} \cup \{[\bar v_i\;\mathbf0]:i\in[k]\}$ attain the optimal error. On $A_{11}$ and $A_{11}$, which contain $0$'s on the diagonal and $1$'s off the diagonal, the vectors $\{[\1\;\mathbf0],[\mathbf0\;\1]\}$ attain zero error on the off-diagonal entries, and $2n$ error in total over the diagonal entries.
Consequently,
\[ \fsnorm{A-A_{2k+2}}\leq t\cdot\fsnorm{B-B_1}+2n , \]
where $t=2n/(kr)$ is the number of copies of $B$ embedded in $A$.

Suppose we have an algorithm that given $A$, returns a rank-$(2k+2)$ matrix $A'$ that satisfies~\Cref{e:upper}.
Let $A_1',\ldots,A_t'$ denote the restriction of $A'$ to the blocks matching the copies of $B$ embedded in $A$.
Thus, $\fsnorm{A-A'} \geq \sum_{i=1}^t\fsnorm{B-A_i'}$.
Furthermore, since $\fsnorm{A}=\Theta(n^2)=\Theta(tnkr)$ and $\fsnorm{B}=\Theta(nkr)$, we have $\fsnorm{A}=O(1)\cdot t\cdot\fsnorm{B}$. Putting everything into~\Cref{e:upper},
\[ \sum_{i=1}^t\fsnorm{B-A_i'} \leq t\cdot\left(\fsnorm{B-B_k} + O(n) + O(\epsilon)\cdot\fsnorm{B}\right) .\]
By averaging, for at least one $i\in[t]$ we have $\fsnorm{B-A_i'}\leq\fsnorm{B-B_k}+O(\epsilon)\cdot\fsnorm{B}$. By scaling $\epsilon$ by a constant, $A'$ solves the rank-$k$ approximation problem for $B$. By the proof for the asymmetric case, this requires $\Omega(nk/\epsilon)$ queries to its input.

\subsubsection{$k=1$}
The previous section proves hardness for rank-$k$ approximation of symmetric distance matrices, for $k\geq4$.
For completeness let us also show hardness for the $k=1$ case, by a somewhat more refined analysis of the reduction in that case.


To this end we slightly modify the construction of $A$ from the previous section.
We draw $B\in\R^{n\times r}$ from the hard distribution for asymmetric distance matrices in the $k=1$ case.
$A$ is constructed as above, except that in $A_{11}$ and $A_{22}$, we change the off-diagonal entries from $1$ to $1.5^2=2.25$.
Again we can scale everything by a constant so that all entries are in the interval $[1,2]$, which yields a distance matrix.

Let $B_1=uv^T$ be the best rank-$1$ approximation of $B$, where $u\in\R^n$ and $v\in\R^r$.
As above we let $\bar v$ denote $n/r$ copies of $v$ concatenated to form a vector in $\R^n$.
Consider the rank-$1$ approximation of $A$ given by $A'=[\bar v\;u][\bar v\;u]^T$.
The error on $A_{12}$ and $A_{21}$ is optimal by construction.
We need to show that the error on $A_{11}$ and $A_{22}$ is at most $\epsilon\cdot O(n^2)$.
To this end we use the fact in our hard distribution that generated $B$ in~\Cref{sec:lb}, for all supported $B$, the top left and right singular vectors are nearly the same. Namely, the top-right one is close to $\1$, and the top-left one is close to $1.5\cdot\1$, for any matrix $B$ in the support.

Concretely, consider an entry in $A_{22}$ whose value is $2.25$.
The corresponding entry in $A'$ is $v_iv_j$ for $i\neq j$.
Each $v_i$ is the mean of a uniformly random vector in $\{1,2\}^r$.
Thus it is a scaled binomial random variable with mean $1.5$ and variance $1/(4r)=\frac14\beta\epsilon$.
Furthermore, $v_i$ and $v_j$ are independent. Therefore, the expected squared Frobenius norm error on that entry is
\[
  \E[(2.25-v_iv_j)^2] =
  \Var[v_i v_j] =
  \Var[v_i]\cdot\Var[v_j] = (\tfrac14\beta\epsilon)^2.
\]
Thus the expected total error over all of the $2.25$ entries (of which there are $O(n^2)$) is $\epsilon\cdot O(n^2)$.
The concentration for the $1$-entries is even stronger.
This completes the proof of the symmetric case for $k=1$.

%% file: appendix_k.tex
\section{Lower Bound for General $k$}
In this section we prove the full statement of~\Cref{thm:lb}.
The proof largely goes by reduction to the $k=1$ case, described as follows.
We take $k$ copies (referred to as~\emph{blocks}) of the hard distribution from the $k=1$ case, and concatenate them horizontally into an $n\times(kr)$ matrix (where as previously, $r=\Theta(1/\epsilon)$).
Then, for each block we pick a Hadamard vector, and add it to all columns in that block. This renders the blocks nearly orthogonal, forcing any low-rank approximation algorithm to compute the majority element of most rows in most blocks, thus solving $\Omega(nk/\epsilon)$ instances of random majority, yielding the desired lower bound.
Even though the description is straightforward, the formal proof requires some elaborate technical work, as given in the rest of this section.

Another rather minor difference is that due to having $k$ blocks (which correspond to $k$ clusters of points in the metric space), we cannot add a ``heavy row'' (which would correspond to a very far point), with all entries set to a large $M>0$, to each block as we did in the $k=1$ case.
The reason is that the clusters are close (the distance between every two clusters is at most $2$), so any point which is far from one cluster must be far from all of them. Thus it would sharpen the spectrum separation of the entire matrix, but not of each block separately, which is the effect we wish to achieve (namely, it would increase the top singular value, but not all top-$k$ singular values.).
This is solved by adding $M^2$ light rows instead of a single heavy row.
In a light row, we can set all distances to a given cluster $i\in[k]$ to $2$, and the rest of the distances to $1$.
This makes the corresponding point slightly further from cluster $i$ than from the rest of the clusters.
Over many similar light rows, this yields the desired effect.

\subsection{Hard distribution}
Given $n,k,\epsilon$, let $\beta,C>0$ be constants that will be chosen later. ($\beta$ will be sufficiently small and $C$ sufficiently large.)
Let $r=\beta/\epsilon$, and assume w.l.o.g.~this is an integer by letting $\epsilon$ be sufficiently smaller.
Let $N=(1+C)n$.

Next we use the Walsh construction of Hadamard vectors.
Recall these are vectors with entries in $\{\pm1\}^N$ which are pairwise orthogonal.
Let $v_1,\ldots,v_k\in\R^N$ be $k$ Hadamard column vectors, which are different than all-$1$'s.
We rescale them to have entries in $\{\pm\frac12\}$.
Let $V^i\in\R^{N\times r}$ be made of $r$ copies of $v^i$ concatenated horizontally.

For every $i=1,\ldots,k$ let $S^i\in\{\pm\frac12\}^{n\times r}$ be made of $n$ vertically stacked instances of majority, after a random permutation of the rows.
We complete it to a matrix $\bar S^i\in\{0,\pm\frac12\}^{N\times r}$ by adding all-$0$'s lines at the bottom.
We use $J$ to denote the all-$1$'s matrix of dimensions implied by context.
We form a matrix $\bar A^i\in\R^{N\times r}$ by
\[
  \bar A^i = 2J + V^i + \bar S^i .
\]
We concatenate the $A^i$'s horizontally to obtain a matrix $\bar A\in\R^{N\times kr}$.
This defines the hard distribution over distance matrices $\bar A\in\R^{N\times kr}$.
\begin{claim}
Every supported $\bar A$ is a bipartite distance matrix.
\end{claim}
\begin{proof}
It can checked that all entries of $\bar A$ are in $\{1,1\tfrac12,2,2\tfrac12,3\}\subset[1,3]$, and thus they satisfy the triangle inequality in a bipartite metric.
\end{proof}

Let $\bar B = \bar A-2J$ and
  $\bar B^i = V^i + \bar S^i$.
Moreover, let $B\in\R^{n\times kr}$ denote the restriction of $\bar B$ to its top $n$ rows.
In most of the proof we will actually work with the matrix $\bar B$ instead of $\bar A$ (cf.~Lemma~\ref{lmm:algb} later on).

\subsubsection{Spectral Properties}
\begin{lemma}\label{lmm:spectralb}
Suppose $(kr)^3=O(n)$.
With probability at least $1-e^{-\Omega(n/(kr)^{3/2}}$, each of the top $k$ squared singular values of $\bar B$ is $\Theta(nr)$, and every other squared singular value is $(\tfrac14\pm o(1))n$.
\end{lemma}
\begin{proof}
Similarly to~Lemma~\ref{lmm:rk1}, all squared singular value of the random portion of $\bar B$ (which are the blocks $S^i$) are $(\tfrac14\pm o(1))n$ with high probability. Since $\bar B$ is obtained by adding $k$ rank-$1$ matrices to that random portion, this bound holds for all but the top-$k$ singular values of $\bar B$.

For the upper part of the spectrum, first note that the absolute value of each entry in $\bar B$ is distributed uniformly i.i.d.~in $\{0,1\}$. Thus with high probability, $\fsnorm{\bar B}\geq\tfrac12Cnkr$. By the above, the squared singular values except the top-$k$ sum to at most $(\tfrac14+o(1))nkr$ (since there are $kr-k$ of them). Thus the sum of squares of the top-$k$ singular vectors is at least, say, $\tfrac14Cnkr$.
On the other hand, for every block $i$, if we multiply $\bar B$ on the left by $v^i$ and on the right by the $i$th block indicator, we get $(1-o(1))Cnr$. Since the Hadamard vectors $\{v_i\}$ are orthogonal and the block indicators are orthogonal, $\bar B$ must have at least $k$ singular values whose square is at least $(1-o(1))Cnr$. The lemma follows.
\end{proof}

\subsubsection{Bounds on Low-Rank Approximation}
For every $i\in[k]$ and $j\in[n]$ let us denote by $s^i_j$ the majority instance which is at the $j$th line of $S^i$.
Let $\mu^i_j$ denote its mean.
As in the $k=1$ case, let $\1$ denote the all-$1$'s vector in $\R^r$, and let $\P$ denote the  orthogonal projection on the subspace spanned by it.

\begin{lemma}\label{lmm:errork}
$\fsnorm{\bar B-\bar B_k}+\epsilon\fsnorm{\bar B} \leq \sum_{i=1}^k\sum_{j=1}^n\norm{s^i_j-\mu^i_j\1}_2^2 + 4(1+C)\beta nk$.
\end{lemma}
\begin{proof}
For the first summand, consider the rank-$k$ matrix given by replacing each majority instance $s^i_j$ in $B$ by $\mu^i_j\1$.
For the second summand, note that each entry in $\bar B$ has magnitude at most $2$, thus $\epsilon\fsnorm{\bar B}\leq \epsilon\cdot(1+C)n\cdot kr = 4(1+C)\beta nk$.
\end{proof}

\begin{corollary}\label{cor:errork}
$\fsnorm{\bar B-\bar B_k}+\epsilon\fsnorm{\bar B} \leq \tfrac14nkr + 4(1+C)\beta kn$.
\end{corollary}
\begin{proof}
In the term $\sum_{i=1}^n\norm{s_i-\mu_i\1}_2^2$ in the above lemma, if we replace $\mu_i$ by $1.5$ (which does not decrease the term since the means are optimal for it), we pay exactly $(\tfrac12)^2$ per entry.
\end{proof}

\subsection{Invoking the Algorithm}
Suppose we have a deterministic algorithm that given $\bar A$, returns $\bar A'$ of rank $k+1$that with probability at least $(2/3)+\delta$ satisfies
\begin{equation}
\fsnorm{\bar A-\bar A'}\leq\fsnorm{\bar A-\bar A_{k+1}}+\epsilon\fsnorm{\bar A} . 
\end{equation}
This the hypothesis of~\Cref{thm:lb}, except with $k+1$ instead of $k$; this will be more convenient to work with, and does not change the theorem statement (one can shift $k$ by $1$ everywhere).

We now move from working with $\bar A$ to working with $\bar B$.
\begin{lemma}\label{lmm:algb}
We can obtain an approximation $\bar B'$ of rank $k+2$ that satifies
\begin{equation}\label{eq:algb}
\fsnorm{\bar B-\bar B'}\leq\fsnorm{\bar B-\bar B_k}+O(\epsilon)\cdot\fsnorm{\bar B} . 
\end{equation}
\end{lemma}
\begin{proof}
We take $\bar B'=\bar A'-2J$. Then, for a constant $c>0$,
\begin{align*}
  \fsnorm{\bar B - \bar B'} &= \fsnorm{(\bar A-2J) - (\bar A'-2J)} \\
  &= \fsnorm{\bar A - \bar A'} \\
  &\leq \fsnorm{\bar A-\bar A_{k+1}}+\epsilon\fsnorm{\bar A} \\
  &\leq \fsnorm{\bar A-\bar B_k-2J}+\epsilon\fsnorm{\bar A} \\
  &\leq \fsnorm{\bar B-\bar B_k}+c\epsilon\fsnorm{\bar B} ,
\end{align*}
and we scale $\epsilon$ down by $c$.
\end{proof}

Combined with Corollary~\ref{cor:errork}, we get
\begin{corollary}\label{cor:errorkb}
$\fsnorm{\bar B-\bar B'}\leq \tfrac14nkr + 4(1+C)\beta kn$.
\end{corollary}

Recall that $\bar B_k'$ denotes the optimal rank-$k$ approximation of $\bar B'$, and thus $\fsnorm{\bar B_k'}$ is the sum of squares of the top-$k$ singular values of $\bar B'$ (which has a total of $k+2$ singular values).

\begin{lemma}\label{lmm:spectralbprime}
The singular values of $\bar B'$ satisfy
\[
  \sum_{i=1}^{k+2}\left(\sigma_i(\bar B)-\sigma_i(\bar B')\right)^2 \leq O(C\beta nk) .
\]
Furthermore, the two bottom squared singular values are each $O(C\beta nk)$.
\end{lemma}
\begin{proof}
Using~Lemma~\ref{lmm:algb} as an upper bound and the Hoffman-Weilandt inequality as a lower bound on $\fsnorm{\bar B-\bar B'}$,
\[ \sum_{i=1}^{k+2}\left(\sigma_i(\bar B)-\sigma_i(\bar B')\right)^2+ \fsnorm{\bar B-\bar B_{k+2}} \leq \fsnorm{\bar B-\bar B'}\leq\fsnorm{\bar B-\bar B_k}+O(C\beta nk) . \]
Observe that $\fsnorm{\bar B-\bar B_k} - \fsnorm{\bar B-\bar B_{k+2}} = \sigma_{k+1}(\bar B)^2 + \sigma_{k+2}(\bar B)^2$, and by Lemma~\ref{lmm:spectralb} each of these two summands is $O(n)$, which is less than $O(C\beta nk)$.\footnote{We recall that $\beta$ and $C$ are constants that will eventually be chosen such that $C\beta$ is smaller than a sufficiently small constant. It holds that $n=O(C\beta nk)$ is $k$ is larger than a sufficiently large constant, which we can assume w.l.o.g.~since we have already proven the $k=1$ case.}
Plugging this above yields the desired inequality, $\sum_{i=1}^{k+2}\left(\sigma_i(\bar B)-\sigma_i(\bar B')\right)^2 \leq O(C\beta nk)$.

As for the bottom two singular values of $\bar B'$, the inequality just proven yields in particular $\left(\sigma_{k+1}(\bar B)-\sigma_{k+1}(\bar B')\right)^2 \leq O(C\beta nk)$, hence $\sigma_{k+1}(\bar B')\leq \sigma_{k+1}(\bar B)+O(\sqrt{C\beta nk})$.
As already mentioned above, $\sigma_{k+1}(\bar B)=O(\sqrt{n})=O(\sqrt{C\beta nk})$ by Lemma~\ref{lmm:spectralb}. Thus $\sigma_{k+1}(\bar B')^2\leq O(C\beta nk)$. The same holds for $\sigma_{k+2}(\bar B')$.
\end{proof}

%

\subsubsection{Averaging Columns in Blocks}\label{sec:kto1}
We now carry out the main part of the reduction to the rank-$1$ case.
We do this by showing that $\bar B$ can be approximated by the matrix resulting from taking $\bar B'$ and replacing each column in each block by the average of columns in that block. Note that in the resulting matrix, each block has rank-$1$ since its columns are identical. Therefore, by averaging, we could get a rank-$1$ approximation for a large constant fraction of the blocks.
Let us now argue this formally.

Let $\Pi_{\1}$ be the orthogonal projection of $\R^{kr}$ on the subspace spanned by block indicators.
Note that for a matrix $Z\in\R^{N\times kr}$, the operation $Z\Pi_{\1}$ averages the columns in each block.
The main lemma for this part of this following.
\begin{lemma}\label{lmm:toblocks}
$\fsnorm{\bar B - \bar B'\Pi_{\1}} \leq \fsnorm{\bar B - \bar B'} + O(\sqrt{C\beta}\cdot nk)$.
\end{lemma}

The proof will go by showing that the row space of $\bar B'$ has to be close to the span of the block indicators, which is the subspace on which $\Pi_{\1}$ projects. (This would yield $\bar B' \approx \bar B'\Pi_{\1}$ and hence $\fsnorm{\bar B - \bar B'} \approx \fsnorm{\bar B - \bar B'\Pi_{\1}}$, as the lemma asserts).
The way we show this is by transitivity, by showing that both subspaces are close to the top-$k$ row space of $\bar B$.
We will require the following technical linear algebraic claims, whose proofs are deferred to~\Cref{sec:kto1proofs} for better readability.

\begin{lemma}\label{lmm:tech1}
Let $A,B\in\R^{n\times m}$.
Let $B=U\Sigma V^T$ be the SVD of $B$.
Suppose $\fsnorm{A-B}\leq\fsnorm{A}-\Delta$.
Then $\fsnorm{AV}\geq\Delta$.
\end{lemma}

\begin{lemma}\label{lmm:tech2}
Let $A\in\R^{n\times n}$ and let $A=U\Sigma V^T$ be its SVD.
Let $V_k$ be the restriction of $V$ to the top-$k$ right singular vectors of $A$.
Let $\Pi$ an orthogonal projection on some $k$-dimensional subspace.
If
\[ (1-\epsilon)k \leq \fsnorm{V_k^T\Pi} \leq k, \]
then
\[ \fsnorm{(A\Pi^\perp)_k} \leq \fsnorm{A_{\epsilon k}} + \fsnorm{(A_{n-k})_k} . \]
(Recall again that $\fsnorm{X_k}$ denotes the sum of squared top-$k$ singular values for every matrix $X$.)
\end{lemma}
As a small digression, let us preview that we will use this lemma twice, on the matrices $\bar B$ and $\bar B'$.
In both cases the projection would be $\Pi_{\mathbf1}$.
In the former case the bound yielded by the lemma would be $O(nk)$, and it the latter case it would be (the better bound) $O(C\beta nk)$.
That is, we will get $\fsnorm{(\bar B\Pi^\perp)_k} = O(nk)$ and $\fsnorm{(\bar B'\Pi^\perp)_k} = O(C\beta nk)$ (see~\Cref{sec:kto1proofs} for an elaboration why).
However we still need to establish the condition $(1-\epsilon)k \leq \fsnorm{V_k^T\Pi} \leq k$ for both invocations, which we will do shortly.

\begin{lemma}\label{lmm:tech3}
Let $V,U,W\in\R^{n\times k}$ matrices such that each has orthonormal columns.
Suppose $\fsnorm{V^TU}\geq(1-\epsilon)k$ and $\fsnorm{U^TW}\geq(1-\epsilon)k$.
Then $\fsnorm{V^TW}\geq(1-O(\epsilon))k$.
\end{lemma}

We now prove~Lemma~\ref{lmm:toblocks}.
Let $\bar B = U\Sigma V^T$ denote the SVD of $\bar B$.
Write it as $\bar B = U\Sigma V^T = U_T\Sigma_TV_T^T + U_B\Sigma_BV_B^T$ where $\Sigma_T$ are the top $k$ singular values and $\Sigma_B$ are the remaining (bottom) singular values.

\begin{lemma}\label{lmm:twice}
Let $\bar B^*\in\R^{N\times kr}$ be a rank-$k'$ matrix such that $\fsnorm{\bar B - \bar B^*} \leq \tfrac14nkr + O(nk)$.
Let $W\in\R^{k'\times kr}$ be an orthonormal basis for the row span of $\bar B^*$.
Then $\fsnorm{WV_T}\geq k-O(\epsilon(k+k'))$.
\end{lemma}
\begin{proof}
On one hand, since $\fsnorm{\bar B}\geq\tfrac12nkr-O(kn)$, the hypothesis $\fsnorm{\bar B - \bar B^*} \leq \tfrac14nkr + O(kn)$ implies, by~Lemma~\ref{lmm:tech1}, $\fsnorm{\bar BW^T} \geq \tfrac14nkr - O(kn)$.
On the other hand,
\begin{align*}
  \tfrac14nkr-O(kn) \leq \fsnorm{\bar BW^T} &= \fsnorm{U_T\Sigma_TV_T^TW^T + U_B\Sigma_BV_B^TW^T} \\
  &= \fsnorm{\Sigma_TV_T^TW^T} + \fsnorm{\Sigma_BV_B^TW^T} \\
  &\leq \norm{\Sigma_T}_2^2\fsnorm{V_T^TW^T} + \norm{\Sigma_B}_2^2\fsnorm{V_B^TW^T} \\
  &\leq (\tfrac14nr+O(n))\cdot\fsnorm{V_T^TW^T} + O(n)\cdot\fsnorm{V_B^TW^T} \\
  &\leq (\tfrac14nr+O(n))\cdot\fsnorm{V_T^TW^T} + O(n)\cdot k' .
\end{align*}
The lemma follows by rearranging and recalling that $r=\Theta(1/\epsilon)$.
\end{proof}

We apply the above lemma twice: once with $\bar B^*$ being $\bar B'$ (whose rank is $k+2$), and once with $\bar B^*$ being the matrix obtained from averaging the columns in each block of $\bar B$ (note that this matrix has rank is $k$).
Since $\Pi_{\1}$ is the orthogonal projection on the row space of that matrix, then by the latter application of Lemma~\ref{lmm:twice} we have
\begin{equation}\label{eq:auxsoon}
  k \geq \fsnorm{V_T^T\Pi_{\1}} \geq k-O(\epsilon k) ,
\end{equation}
which establishes the condition of Lemma~\ref{lmm:tech2}, yielding
\[
  \fsnorm{(\bar B\Pi_{\1}^\perp)_k} \leq O(nk) . 
\]
For the former application, let $\bar B'=U'\Sigma'(V')^T$ denote the SVD of $\bar B'$.
Corollary~\ref{cor:errorkb} provides the requirement of Lemma~\ref{lmm:twice}, which in turn yields $\fsnorm{(V')^TV_T}\geq k-O(\epsilon k)$.
Together with Equation~(\ref{eq:auxsoon}), by transitivity (Lemma~\ref{lmm:tech3}),
\[
  k \geq \fsnorm{(V')^T\Pi_{\1}} \geq k-O(\epsilon k), 
\]
which establishes the condition of Lemma~\ref{lmm:tech2}, yielding
\[
  \fsnorm{(\bar B'\Pi_{\1}^\perp)_k} \leq O(C\beta nk) . 
\]
Together with the above,
\begin{equation}\label{eqpi1}
\norm{(\bar B\Pi_{\1}^\perp)_k}_F\norm{(\bar B'\Pi_{\1}^\perp)_k}_F \leq O(\sqrt{C\beta}\cdot nk) .
\end{equation}
We can extend this from rank-$k$ to rank-$(k+2)$ since the additional two square singular value of each of the matrices is $O(C\beta nk)$ (cf.~Lemmas~\ref{lmm:spectralb}~and~\ref{lmm:spectralbprime}).\footnote{Recall again that we set $C\beta<1$.}
Since $\bar B'\Pi_{\1}^\perp$ has rank $k+2$, then by Hoffman-Weilandt,
\begin{align*}
  \fsnorm{\bar B\Pi_{\1}^\perp - \bar B'\Pi_{\1}^\perp} &\geq \sum_{i}\left(\sigma_i(\bar B\Pi_{\1}^\perp) - \sigma_i(\bar B'\Pi_{\1}^\perp)\right)^2 & \text{by Hoffman-Weilandt} \\
  &= \fsnorm{\bar B\Pi_{\1}^\perp} - \sum_{i=1}^{k+2}\sigma_i(\bar B\Pi_{\1}^\perp)\sigma_i(\bar B'\Pi_{\1}^\perp) + \fsnorm{\bar B'\Pi_{\1}^\perp} & \mathrm{rank}(\bar B'\Pi_{\1}^\perp)=k+2 \\
  &\geq \fsnorm{\bar B\Pi_{\1}^\perp} - \sum_{i=1}^{k+2}\sigma_i(\bar B\Pi_{\1}^\perp)\sigma_i(\bar B'\Pi_{\1}^\perp) & \\
  &\geq \fsnorm{\bar B\Pi_{\1}^\perp} - \norm{(\bar B\Pi_{\1}^\perp)_{k}}_F\norm{(\bar B'\Pi_{\1}^\perp)_{k}}_F & \text{by Cauchy-Schwartz} \\
  &\geq \fsnorm{\bar B\Pi_{\1}^\perp} - O(\sqrt{C\beta}\cdot nk) & \text{by~\Cref{eqpi1} .} 
\end{align*}
Finally, by Pythagorean identities,
\begin{align*}
  \fsnorm{\bar B - \bar B'\Pi_{\1}} &= \fsnorm{\bar B\Pi_{\1} - \bar B'\Pi_{\1}} + \fsnorm{\bar B\Pi_{\1}^\perp} \\
  &= \fsnorm{\bar B - \bar B'} + \left(\fsnorm{\bar B\Pi_{\1}^\perp} - \fsnorm{\bar B\Pi_{\1}^\perp - \bar B'\Pi_{\1}^\perp}\right) \\
  &\leq \fsnorm{\bar B - \bar B'} + O(\sqrt{C\beta}\cdot nk) .
\end{align*}
This proves~Lemma~\ref{lmm:toblocks}.

\subsubsection{Relevant Blocks}


\begin{lemma}\label{lmm:ubfine}
There is a subset $I\subset[k]$ of size at least $|I|\geq0.99k$ such that for every $i\in I$,
\begin{equation}\label{eq:ubfine}
\fsnorm{\bar B^i-(\bar B')^i\P} \leq \fsnorm{\bar B^i = (\bar B^i)_1} + O(\sqrt{C\beta}\cdot n) ,
\end{equation}
where $(\bar B^i)_1$ is (as usual) the optimal rank-$1$ approximation of $\bar B^i$.
We refer to blocks $B_i$ with $i\in I$ as~\emph{relevant blocks}.
\end{lemma}
\begin{proof}
By Lemma~\ref{lmm:toblocks}, $\fsnorm{\bar B - \bar B'\Pi_{\1}} \leq \fsnorm{\bar B - \bar B'} + O(\sqrt{C\beta}\cdot nk)$.
Note that the left-hand side equals $\sum_{i=1}^k\fsnorm{\bar B^i-(\bar B')^i\P}$.
As for the right-hand side, by~\Cref{eq:algb} we have $\fsnorm{\bar B - \bar B'} \leq \fsnorm{\bar B - \bar B_k}+O(\epsilon)\cdot\fsnorm{B}$, and we recall that $\fsnorm{B}=O(Cnkr)=O(C\beta nk/\epsilon)$.
Furthermore, $\fsnorm{\bar B - \bar B_k} \leq \sum_{i=1}^k\fsnorm{\bar B^i - (\bar B^i)_1}$. 
Putting it all together yields $\sum_{i=1}^k\fsnorm{\bar B^i-(\bar B')^i\P} \leq \sum_{i=1}^k\fsnorm{\bar B^i - (\bar B^i)_1} + O(\sqrt{C\beta}\cdot nk)$, or rearranging,
\[ \sum_{i=1}^k\left(\fsnorm{\bar B^i-(\bar B')^i\P} - \fsnorm{\bar B^i - (\bar B^i)_1}\right) \leq O(\sqrt{C\beta}\cdot nk) . \]
Each term in the sum on the left-hand side is non-negative, by the optimality of $(\bar B^i)_1$ for rank-$1$ approximation of $\bar B^i$. Therefore we can use an averaging argument (Markov's inequality) and conclude that at least $0.99k$ of the $k$ summands on the left-hand side are at most $100/k$ times the right-hand side. The lemma follows.
\end{proof}

Fix $i\in I$. Since $(\bar B')^i\P$ is a rank-$1$ matrix we can write it as $\bar a^i(b^i)^T$ where $\bar a^i\in\R^N$ and $b^i\in\R^r$.
Let $a^i\in\R^n$ denote the restriction of $\bar a^i$ to the first $n$ entries.
Consider $\fsnorm{\bar B^i-\bar a^i(b^i)^T}$.
Note that the last $Cn$ rows of are either all $1$ or all $-1$, depending on the Hadamard vector $v^i$.
Let $\sigma^i_j\in\{\pm1\}$ denote the sign of row $j$.
Then the contribution of the last $Cn$ rows is $\sum_{j=n+1}^{Cn}\norm{\sigma^i_j\1-a^i_jb^i}_2^2$ which can be rewritten as $\sum_{j=n+1}^{Cn}\norm{\1-\sigma^i_ja^i_jb^i}_2^2$.
Pick the $j$ that minimizes the term $\norm{\1-\sigma^i_ja^i_jb^i}_2^2$ and set all entries $a^i_{n+1},\ldots,a^i_{Cn}$ to $a^i_j$ with the appropriate sign, to obtain a vector $\hat a^i$. By choice of $j$ we have
\begin{equation}\label{eqmnew1}
  \fsnorm{\bar B^i-\hat a^i(b^i)^T} \leq \fsnorm{\bar B^i-\bar a^i(b^i)^T} .
\end{equation}
Furthermore,
\begin{equation}\label{eqmnew2}
  \fsnorm{\bar B^i-\hat a^i(b^i)^T} = \fsnorm{B^i-a^i(b^i)^T} + Cn\norm{b^i-\1}_2^2 .
\end{equation}
Combining~\Cref{eqmnew1,eqmnew2},
\begin{equation}\label{eqmnew}
  \fsnorm{B^i-a^i(b^i)^T} + Cn\norm{b^i-\1}_2^2 \leq \fsnorm{\bar B^i-\bar a^i(b^i)^T} .
\end{equation}
If we use~Lemma~\ref{lmm:ubfine} as an upper bound on $\fsnorm{\bar B^i-\bar a^i(b^i)^T}$ and~Lemma~\ref{lmm:rk1} as a lower bound on $\fsnorm{B^i-a^i(b^i)^T}$, we get $Cn\norm{\1-b^i}_2^2 \leq O(n)$, which rearranges to
\begin{equation}\label{eq:1bk}
  \norm{\1-b^i}_2^2 \leq \frac{O(1)}{C} .
\end{equation}
This implies Lemmas~\ref{lmm:bproj} and~\ref{lmm:a} for every relevant block, by the same proofs as their original proofs.

\subsection{Solving Majority}
Recall we have a total of $nk$ majority instances (each of length $r$) embedded in $\bar B$.
Note by the construction of $\bar B$, each of them has alphabet either $\{0,1\}$ or $\{0,-1\}$, depending on the sign of the corresponding entry of the Hadamard vector $v^i$, where $i$ the block in which the instance is embedded.

By~Lemma~\ref{lmm:antic} and Markov's inequality, at least $0.9nk$ of the instances are typical.
For an instance with alphabet $\{0,1\}$, we solve it using $\bar B'$ by reporting that the majority element is $1$ if the average over the corresponding entries in $\bar B'$ is larger than $0.5$, and reporting $0$ if it is smaller than $0.5$.
Instances with alphabet $\{0,-1\}$ are solved similarly with threshold $-0.5$.
Note that the solution procedure compares the threshold to the mutual value of the corresponding entries of $\bar B'\Pi_\1$.
If we are correct on an instance, we say it is~\emph{solved}, and otherwise~\emph{unsolved}.
For relevant block $i\in I$, let $S_i'$ denote the subset of majority instances which are both typical and unsolved.
Let $S'=\cup_{i\in I}S_i'$ be the subset of all instances which are typical, unsolved, and embedded in a relevant block.

Our goal is to show that we solve each instance with probability at least $2/3$.
Since the instances were placed in $\bar B$ by random permutation, every instance has the same probability $p$ to be solved, thus we need to show $p\geq2/3$. Suppose by contradiction that $p<2/3$.
Since at least $0.9nk$ instances are typical, and at least $0.99k$ blocks are relevant, then there is a fixed constant $\zeta>0$ (this was $\tfrac1{15}$ in the $k=1$ case) such that $|S'|\geq\zeta nk$.

For every $i\in I$ we have (by definition of relevant blocks),
\[
  \fsnorm{\bar B^i - (\bar B'\Pi_\1)^i} 
  \leq \fsnorm{\bar B^i - (\bar B^i)_1} + O(\sqrt{C\beta}\cdot n)
  \leq \sum_{j=1}^n\norm{s_j-\mu_j\1}_2^2 + O(\sqrt{C\beta}\cdot n).
\]

By bounding $\sum_{j=1}^n\norm{s_j-\mu_j\1}_2^2$ in the same way as in the $k=1$ case,
\begin{equation}\label{eq:ubk}
  \fsnorm{\bar B^i - (\bar B'\Pi_\1)^i} \leq
  |S_i'|(\tfrac14r-\gamma^2) + \sum_{s_j\notin S_i'}\norm{s_j-\mu_j\1}_2^2 + O(\sqrt{C\beta}\cdot n).
\end{equation}
Similarly, if we denote $(\bar B')^i\P=\bar a^i(b^i)^T$ (since this is a rank-$1$ matrix), then as in the $k=1$ case,
\begin{equation}\label{eq:lbk}
  \fsnorm{B^i-(B'\Pi_\1)^i} =
  \sum_{i=1}^n\norm{s_i-\bar a_j^i(b^i)^T}_2^2 \geq
  |S_i'|(\tfrac14r + \eta^2 - 2\gamma\eta) + \sum_{s_j\notin S_i'}\norm{s_j-\mu_j\1}_2^2 .
\end{equation}
(We remark that the latter inequality relies on Lemmas~\ref{lmm:bproj} and~\ref{lmm:a}, which were proven in the previous section for relevant blocks, based on~\cref{eq:1bk}; as per Lemma~\ref{lmm:bproj}, $\eta=\Theta(1/\sqrt C)$.)

Together,
\[
   |S_i'|(\tfrac14r + \eta^2 - 2\gamma\eta) \leq |S_i'|(\tfrac14r-\gamma^2) + O(\sqrt{C\beta}\cdot n) .
\]
We sum this over all $i\in I$, and recall that $\sum_{i\in I}|S_i'|=|S'|$. This yields,
\[
   |S'|(\tfrac14r + \eta^2 - 2\gamma\eta) \leq |S'|(\tfrac14r-\gamma^2) + O(\sqrt{C\beta}\cdot kn) ,
\]
which rearranges to $|S'|(\gamma-\eta)^2 \leq O(\sqrt{C\beta}\cdot nk)$. Recalling that $|S'|\geq\zeta nk$, we get $\zeta(\gamma-\eta)^2 \leq O(\sqrt{C\beta}\cdot)$.
Since $\zeta$ and $\gamma$ are fixed constant, we can take $\eta$ and $\beta$ to be sufficiently small, and arrive at the desired contradiction.
\qed

\subsection{Deferred Proofs from Appendix~\ref{sec:kto1}}\label{sec:kto1proofs}
\begin{lemma}
Let $A,B\in\R^{n\times m}$.
Let $B=U\Sigma V^T$ be the SVD of $B$.
Suppose $\fsnorm{A-B}\leq\fsnorm{A}-\Delta$.
Then $\fsnorm{AV}\geq\Delta$.
\end{lemma}
\begin{proof}
\begin{align*}
  \fsnorm{AV} &= \fsnorm{AVV^T} & \\
  &= \fsnorm{A} - \fsnorm{A(I-VV^T)} & \text{Pythagorean theorem} \\
  &\geq \Delta + \fsnorm{A-B} - \fsnorm{A(I-VV^T)}& \\
  &= \Delta + \fsnorm{AVV^T - BVV^T} + \fsnorm{A(I-VV^T)} - \fsnorm{A(I-VV^T)}& \text{Pythagorean theorem} \\
  &\geq \Delta.
\end{align*}
\end{proof}

\begin{lemma}
Let $A\in\R^{n\times n}$ and let $A=U\Sigma V^T$ be its SVD.
Let $V_k$ be the restriction of $V$ to the top-$k$ right singular vectors of $A$.
Let $\Pi$ an orthogonal projection on some $k$-dimensional subspace.
If
\[ (1-\epsilon)k \leq \fsnorm{V_k^T\Pi} \leq k, \]
then
\[ \fsnorm{(A\Pi^\perp)_k} \leq \fsnorm{A_{\epsilon k}} + \fsnorm{(A_{n-k})_k} . \]
(Recall again that $\fsnorm{X_k}$ denotes the sum of squared top-$k$ singular values for every matrix $X$.)
\end{lemma}
\begin{proof}
We have $\fsnorm{(A \Pi^\perp)_k} = \fsnorm{(A_{n-k} \Pi^\perp + A_k \Pi^\perp)_k}$.
We can write $A \Pi^\perp$ as $A_{n-k} \Pi^\perp + A_k \Pi^\perp$, and for any vector $x$,
$A_{n-k} \Pi^\perp x$  and $A_k \Pi^\perp x$ are orthogonal, and so
$\norm{A \Pi^\perp x}_2^2 = \norm{A_{n-k} \Pi^\perp x}_2^2 + \norm{A_k \Pi^\perp x}_2^2$.
It follows that
\[ \fsnorm{(A \Pi^\perp)_k} \leq \fsnorm{A_k \Pi^\perp} + \fsnorm{(A_{n-k} \Pi^\perp)_k} . \] 
Note that $\fsnorm{(A_{n-k} \Pi^{\perp})_k} \leq \fsnorm{(A_{n-k})_k}$.
Thus it remains to show $\fsnorm{A_k \Pi^\perp}\leq \fsnorm{A_{\epsilon k}}$, 
or equivalently, by the Pythagorean theorem, 
$\fsnorm{A_k \Pi} \geq \fsnorm{A_k}-\fsnorm{A_{\epsilon k}} = \sum_{i=\epsilon k+1}^k \sigma_i^2(A)$.

We have $\fsnorm{A_k \Pi} = \fsnorm{\Sigma_k V_k^T \Pi}$ and we know $\fsnorm{V_k^T \Pi} \geq
k(1-\epsilon)$.
Let $R$ be an $n\times n$ rotation matrix that takes $V_k^T$ to $[I_k \;
0]$, where here $I_k$ is the identity matrix of order $k$ and $0$ is an $k\times (n-k)$ zero matrix.
Replace $\Pi$ with $R^T\Pi$.
Then $\fsnorm{A_k \Pi} = \fsnorm{\Sigma_k (V_k^T R)(R^T \Pi)}$ and $\fsnorm{(V_k^T R)(R^T \Pi)} \geq k (1-\epsilon)$.
Thus, we can assume w.l.o.g.~that $V_k^T = [I_k \; 0]$, and so $R^T \Pi$ has the form $[\Phi\; 0]$,
where $\Phi$ is $k \times k$.

Thus $\fsnorm{A_k \Pi} = \fsnorm{\Sigma_k \Phi}$ subject to $\fsnorm{\Phi}\geq k(1-\epsilon)$.
Also each row of $\Phi$ has squared norm at most $1$ since it is a submatrix of a rotation matrix. Consequently, since $\Sigma_k$ is a diagonal matrix, $\fsnorm{\Sigma_k \Phi}$ is minimized when placing all mass of $\Phi$ on the bottom $k(1-\epsilon)$ rows, and in this case it is exactly $\sum_{i=\epsilon k+1}^k \sigma_i^2(A)$.
\end{proof}

We have applied this lemma in Appendix~\ref{sec:kto1} to both $\bar B_k$ and $\bar B'$.
Let us show the resulting upper bound $\fsnorm{A_{\epsilon k}} + \fsnorm{(A_{n-k})_k}$ in each case.

For $\bar B_k$, by Lemma~\ref{lmm:spectralb} we know that the top $k$ squared singular values of $\bar B$ are $\Theta(Cnr)=O(Cn\beta/\epsilon)$ each, thus $\fsnorm{\bar B_{\epsilon k}}= O(\epsilon k \cdot Cn\beta/\epsilon)= O(C\beta nk)$, and the rest of the squared singular values are $\Theta(n)$, thus $\fsnorm{(\bar B_{n-k})_k}=O(kn)$. The total bound is $O(nk)$.

For $\bar B'$, 
\begin{align*}
  \fsnorm{\bar B_{\epsilon k}'} &= \sum_{i=1}^{\epsilon k}\sigma_i(\bar B')^2 & \\
  &= \sum_{i=1}^{\epsilon k}\left(\sigma_i(\bar B')-\sigma_i(\bar B)+\sigma_i(\bar B)\right)^2 & \\
  &\leq \sum_{i=1}^{\epsilon k}2\left(\sigma_i(\bar B)^2+(\sigma_i(\bar B)-\sigma_i(\bar B'))^2\right) & \text{Similarly to Claim~\ref{clm:triangle}} \\
  &= 2\fsnorm{\bar B_{\epsilon k}} + 2\sum_{i=1}^{\epsilon k}(\sigma_i(\bar B)-\sigma_i(\bar B'))^2 . \\
\end{align*}
The first term was already upper bounded by $O(C\beta nk)$ above, and the second sum is upper bounded by $O(C\beta nk)$ by Lemma~\ref{lmm:spectralbprime}.
The term $\fsnorm{(\bar B_{n-k}')_k}$ is $O(C\beta nk)$ since there are two remaining eigenvalues and each is $O(C\beta nk)$ by Lemma~\ref{lmm:spectralbprime}. The total bound is $O(C\beta nk)$.

%

\begin{lemma}
Let $V,U,W\in\R^{n\times k}$ matrices such that each has orthonormal columns.
Suppose $\fsnorm{V^TU}\geq(1-\epsilon)k$ and $\fsnorm{U^TW}\geq(1-\epsilon)k$.
Then $\fsnorm{V^TW}\geq(1-O(\epsilon))k$.
\end{lemma}
\begin{proof}
Note we can replace $V^T$ with $V^TR$ and $U$ with $R^TU$ and $W$ with $R^TW$, where $R$ is an $n\times n$ rotation which takes $U$ to the top $k$ standard unit vectors.
All norms in the premise and goal of the claim are preserved.
So we can assume that $\fsnorm{V_{top}} \geq k(1-\epsilon)$, $\fsnorm{W_{top}} \geq  k(1-\epsilon)$, and need to show $\fsnorm{V^TW} \geq k - O(k\epsilon)$, where ``top`` means the top $k\times k$ submatrix with remaining rows replaced with $0$s. Let $W = W_{top} + W_{rest}$ and $V = V_{top} + V_{rest}$. 

Then,
\begin{align*}
 \fsnorm{V^TW} &= \fsnorm{V_{top}^T W_{top} + V_{rest}^T W_{rest}} & \\
&\geq \fsnorm{V_{top}^T W_{top}} - 2Tr(W_{rest}^T V_{rest} V_{top}^T W_{top}) & \text{(i)} \\
&= \fsnorm{V_{top}^T W_{top}} - 2Tr (W_{top} W_{rest}^T V_{rest} V_{top}^T) & \text{(ii)} \\
&\geq \fsnorm{V_{top}^T W_{top}} - 2 \norm{W_{top} W_{rest}}_F \norm{V_{rest} V_{top}^T}_F & \text{(iii)} \\
&\geq \fsnorm{V_{top}^T W_{top}} - 2 \norm{W_{top}}_2\norm{W_{rest}}_F\norm{V_{top}}_2\norm{V_{rest}}_F & \text{(iv)} \\
&\geq \fsnorm{V_{top}^T W_{top}} - 2 \norm{W_{rest}}_F \norm{V_{rest}}_F & \text{(v)} \\
&= \fsnorm{V_{top}^T W_{top}} - 2 (1-\fsnorm{W_{top}})^{1/2} (1-\fsnorm{V_{top}})^{1/2} & \\
&= \fsnorm{V_{top}^T W_{top}} - 2 (\epsilon k)^{1/2} (\epsilon k)^{1/2}  & \\
&= \fsnorm{V_{top}^T W_{top}} - 2\epsilon k ,  &   (*)
\end{align*}
where,
\begin{itemize}
  \item (i) is by expanding the square and dropping a non-negative term;
  \item (ii) is by cyclicity of trace;
  \item (iii) is since $Tr(AB) <= \norm{A}_F \norm{B}_F$;
  \item (iv) is by submultiplicativity of operator and Frobenius norm;
  \item (v) is since $W$ and $V$ are orthonormal so their operator norm is $1$, and operator norm does not decrease by taking submatrices.
\end{itemize}

So we just need to lower bound $\fsnorm{V_{top}^T W_{top}}$, and we can drop the last $n-k$ rows of $V_{top}$ and $W_{top}$ since they are zeros.
Next, we write $V_{top}$ in its SVD, $V_{top}= A \Sigma B^T$.
Then since $\fsnorm{V_{top}} \geq k(1-\epsilon)$ and $\norm{V_{top}}_2^2 \leq 1$ (since it is a submatrix of $V$), necessarily, there are at least $k(1-2\epsilon)$ singular values of squared value at least $1-2\epsilon$.
Indeed, otherwise $\fsnorm{V_{top}} \leq k(1-2\epsilon)(1-2\epsilon) + 2\epsilon k \cdot 1 < k(1-\epsilon)$ for $\epsilon$ less than a small enough constant.
Let $\Sigma_h$ be these singular values and $\Sigma_l$ be the remaining ones.
Then
\[ \fsnorm{V_{top}^T W_{top}} = \fsnorm{\Sigma B^T W_{top}} \geq (1-2\epsilon) \fsnorm{B_h^T W_{top}} . \]
Now $\fsnorm{W_{top}} \geq k(1-\epsilon)$, and $B_h$ is a $k(1-2\epsilon)$-dimensional subspace of $\text{span}(e_1, ..., e_k)$ (the standard unit vectors), and so we can extend it with an orthonormal basis $B'$ so that $\text{span}(B_h, B') = \text{span}(e_1, ..., e_k)$.
Then by the Pythagorean theorem
\[ k(1-\epsilon) \leq \fsnorm{W_{top}} = \fsnorm{B_h^T W_{top}} + \fsnorm{B' W_{top}} ,\]
and since $\norm{W_{top}}_2^2 \leq 1$, we have $\fsnorm{B' W_{top}} \leq \fsnorm{B'} \leq 2\epsilon k$.
Consequently, $\fsnorm{B_h^T W_{top}} \geq k(1-\epsilon) - 2 \epsilon k = k- 3\epsilon k$.
Hence, $\fsnorm{V_{top}^T W_{top}} \geq (1-2\epsilon) k (1-3\epsilon) \geq k (1-O(\epsilon))$.
Plugging into (*) gives us our desired $k (1-O(\epsilon))$ lower bound on $\fsnorm{V^TW}$. 
\end{proof}